\documentclass[conference]{IEEEtran}
\IEEEoverridecommandlockouts
\IEEEpubidadjcol
\usepackage[left=0.625in,right=0.655in,top=0.705in,bottom=1.03in]{geometry}
\usepackage[utf8]{inputenc} 
\usepackage[T1]{fontenc}
\usepackage{url}
\usepackage{ifthen}
\newboolean{isarxiv}
\setboolean{isarxiv}{false}
\newcommand{\arxivonly}[1]{\ifthenelse{\boolean{isarxiv}}{#1}{}}
\newcommand{\confonly}[1]{\ifthenelse{\boolean{isarxiv}}{}{#1}}
\usepackage{cite}
\usepackage[cmex10]{amsmath}
\usepackage{siunitx-v2}
\usepackage{amsfonts}
\usepackage{amssymb}
\usepackage{stfloats} 
\usepackage{upgreek}

\usepackage{amsthm}
\usepackage{bm}
\usepackage{algorithmic}
\usepackage{graphicx}
\usepackage{textcomp}
\usepackage{xcolor}
\definecolor{airforceblue}{rgb}{0.36,
0.54, 0.66}
\definecolor{amber}{rgb}{1.0, 0.75,
0.0}
\definecolor{amethyst}{rgb}{0.6, 0.4,
0.8}
\definecolor{applegreen}{rgb}{0.55,
0.71, 0.0}
\definecolor{antiquebrass}{rgb}{0.8,
0.58, 0.46}
\usepackage{enumitem}
\usepackage{mathtools}
\PassOptionsToPackage{hyphens}{url}
\usepackage{tikz,pgfplots}
\usetikzlibrary{spy}
\usetikzlibrary{positioning}
\usetikzlibrary{decorations.pathreplacing}
\usetikzlibrary {arrows.meta,bending,positioning}
\pgfdeclareradialshading{sphere}{\pgfpoint{0cm}{0cm}}%
{rgb(0cm)=(1,1,1);
rgb(0.7cm)=(0.7,0.1,0); rgb(1cm)=(0.5,0.05,0); rgb(1.05cm)=(1,1,1)}
\usepackage{grffile}
\pgfplotsset{compat=newest}
\usepgfplotslibrary{patchplots}
\usetikzlibrary{backgrounds,plotmarks,fit,positioning,arrows,shapes,shapes.multipart,calc,arrows.meta}
\usetikzlibrary{fit,positioning,arrows,shapes,shapes.multipart,calc,arrows.meta,shapes.geometric,shapes.misc}
\usetikzlibrary{decorations.pathreplacing,angles,quotes,calligraphy}
\usetikzlibrary{automata}

\allowdisplaybreaks

\usepackage[caption=false,font=normalsize,labelfont=sf,textfont=sf]{subfig}

\usepackage{vmr-symbols-vecbold}
\usepackage{standard-macros}

\usepackage{soul} 
\sethlcolor{yellow} 

\newtheorem{definition}{Definition}
\newtheorem{lemma}{Lemma}
\newtheorem{remark}{Remark}
\newtheorem{claim}{Claim}

\renewcommand{\define}{=}
\newcommand{\n}{{\uN}} 
\newcommand{\Ka}{{\uK_{\rm a}}} 
\newcommand{\Kaf}{{\uK_{\rm a1} }} 
\newcommand{\Kas}{{\uK_{\rm a2}}}

\newcommand{\B}{{\uB}}
\newcommand{\power}{\uP} 
\newcommand{\M}{\uM} 

\newcommand{\Ber}[1]{{\rm Ber}\lefto(#1\right)}

\newcommand{\msC}{\mathscr{C}}\newcommand{\of}[1]{^{(#1)}}

\newtheorem{theorem}{Theorem}
\newcommand{\zv}{\veczero}

\newcommand{\rvY}{\rvecy}

\def\BibTeX{{\rm B\kern-.05em{\sc i\kern-.025em b}\kern-.08em
    T\kern-.1667em\lower.7ex\hbox{E}\kern-.125emX}}

\title{Minimum Energy per Bit of Unsourced Multiple Access with Location-Based Codebook Partitioning
\vspace{-.3cm}}

\author{
\IEEEauthorblockN{Deekshith Pathayappilly Krishnan$^{*}$, Kaan Okumus$^{*}$, Khac-Hoang Ngo$^{\dagger}$, and Giuseppe Durisi$^{*}$
}
\IEEEauthorblockA{
$^{*}$Department of Electrical Engineering, Chalmers University of Technology, 41296 Gothenburg, Sweden\\ 
$^{\dagger}$Department of Electrical Engineering, Link\"oping University, 58183 Link\"oping, Sweden\\ Email: \{deepat, okumus\}@chalmers.se, khac-hoang.ngo@liu.se, durisi@chalmers.se\\
\vspace{-2em}
}
\thanks{This work was supported in part by the Swedish Research Council under
grants 2021-04970 and 2022-04471, and by the Swedish Foundation for
 Strategic Research. The  work of K.-H. Ngo was supported in part by the Excellence Center at Linköping--Lund in Information Technology (ELLIIT).}
}%

\begin{document}

\maketitle

\begin{abstract}
   We derive finite-blocklength bounds on the minimum achievable energy per bit over a Gaussian unsourced multiple access (UMA) channel in the presence of heterogeneous path-loss conditions. 
   We consider a setting in which the path loss is known to the users, which enables the use of location-based codebook partitioning [\c{C}akmak et al., 2025].
   Through numerical simulations and a large-system analysis based on the replica method, we quantify the performance gain of this strategy relative to the conventional UMA approach in which all users employ a common codebook.
\end{abstract}

\section{Introduction}
Massive machine-type communication (mMTC) is a rapidly growing use case in wireless networks
due to the fast expansion of Internet of Things (IoT) systems. {mMTC} is characterized by sporadic, uplink-centered
transmissions of short packets from a high density of devices often operating under stringent
energy-efficiency requirements~\cite{kalor2024wireless}.
As first formalized in~\cite{PolyanskiyISIT2017massive_random_access}, the {mMTC} problem has
unique features, which make it different from the traditional multiple access channel (MAC). One key difference is that, to capture the massiveness of the user population, it is fundamental to assume that
it is impossible for the system to assign to each user a different codebook.
The resulting scenario, in which all users are equipped
with the same codebook, is commonly referred to in the literature as unsourced multiple access (UMA).

Finite-blocklength bounds on the minimum energy per bit for the Gaussian {UMA} scenario were first obtained
in~\cite{PolyanskiyISIT2017massive_random_access}, under the assumption
that all users are received at the same power. These bounds have been improved and generalized to different scenarios, including settings where the number of active users is random and unknown to the receiver~\cite{Ngo2023_randomActivity}, as well as quasi-static fading channels~\cite{Kowshik2020}. See~\cite{andreev25-a} for a recent review of the area.
Furthermore, coding schemes approaching these bounds have been designed.
A particularly promising
approach~\cite{fengler2019sparcs,Amalladinne2018coupled_compressed_sensing,Amalladinne2020,Amalladinne2022} leverages the similarity of {UMA} decoding and sparse signal recovery and relies on approximate message passing (AMP) for message recovery.
This strategy, however, incurs a challenge when applied to realistic channel models in which signals
from different users are received at different average-power levels because of path-loss effects.
Indeed, since in {UMA} it is, in general, not possible to establish any association between the
position of a user (and, hence, its path loss) and the codeword the user will transmit, one is forced
to use at the {AMP} denoiser a \emph{diffuse} prior, which involves an averaging over all possible user positions.

This problem has been recently sidestepped in~\cite{Cakmak2025} via the insightful observation that, in
most cellular wireless networks, an estimate of the path loss between users and base-stations can be obtained
at the users via the downlink control information transmitted by the cellular base stations.
This information can be used in a {UMA} system as follows: in the system-design phase, one quantizes the possible
path loss values using a pre-determined number $Q$ of levels; then, one partitions the codebook in $Q$ subcodebooks.
In the operational phase, each user estimates its path loss, quantizes it, and uses the corresponding subcodebook to transmit its message.
This strategy, referred to in~\cite{Cakmak2025} as \emph{location-based codebook partitioning}, allows one to
establish a link between transmitted codewords and path loss, and results in a much more concentrated
prior for the {AMP} denoiser.
As shown in~\cite{Cakmak2025} both theoretically and experimentally, this yields better
{AMP} decoder performance.

\paragraph*{Contributions}
The purpose of this paper is to assess the effectiveness of location-based codebook partitioning via
finite-blocklength information-theoretic bounds similar to the ones developed
in~\cite{PolyanskiyISIT2017massive_random_access}.
For simplicity, we restrict our attention to a toy-model version of the scenario considered
in~\cite{Cakmak2025}.
Specifically, we assume that the channel between each user and a single-antenna receiver is modeled
as a nonfading Gaussian channel, with (deterministic) path loss that can take only two different values,
which we denote by $g_{1}$ and $g_{2}$, respectively.

For this setup, we derive finite-blocklength bounds on the minimum energy per bit achievable when
the codebook is partitioned into two subcodebooks and each user selects the subcodebook corresponding
to its path loss. Furthermore, we compare this achievability bound with those obtained when all users select codewords from the same codebook, and when the
decoder employs successive interference cancellation.
Through numerical results, we show that the bound on the minimum energy per bit achieved via
location-based codebook partitioning lies below the two alternative bounds.
We also present simulation results obtained by combining the coded
compressive sensing (CCS) scheme proposed in~\cite{Amalladinne2020}
with the multisource-{AMP} decoder introduced in~\cite{Cakmak2025}. 
These results exhibit a similar performance ordering. 

To provide insights into the benefits of location-based codebook partitioning, we
finally present a large-system characterization of the per-user probability of error (PUPE) via replica
analysis~\cite{edwards1975theory}, in the regime where both the number of users and
the blocklength grow to infinity with a fixed ratio.
This (non-rigorous) large-system analysis reveals that, with location-based codebook
partitioning, the {UMA} large-system performance can be characterized by analyzing two
equivalent scalar Gaussian channels with SNR proportional to $g_{1}^{2}$ and $g_{2}^{2}$,
respectively.
On the contrary, when a single codebook is used, the relevant scalar channel is a fading
channel, with instantaneous fading value (not known to the receiver), taking value in $\{g_{1},g_{2}\}$
with probability depending on the asymptotic fraction of users experiencing path loss $g_{1}$ and
$g_{2}$, respectively.

\subsubsection*{Notation}
We denote system parameters by uppercase non-italic
letters, e.g.,~$\M$. Uppercase italic letters, e.g.,~$X$, denote scalar random variables and their realizations are in lowercase, e.g.,~$x$. Vectors are denoted likewise in boldface, e.g., a random vector~$\rvecx$ and its realization~$\vecx$. We denote their $i$th entries as~$[\rvecx]_i$ and $[\vecx]_i$.  We use a script font for random matrices, e.g.,~$\msC$, and a sans-serif font for deterministic matrices, e.g.,~$\matC$. 
We denote the $n \times n$ identity matrix by~$\matidentity_n$, and the all-zero vector by $\veczero$.
The superscript~$\tp{}$ 
stands for transposition. 
We denote sets with calligraphic letters, e.g., $\setS$, the set $\set{1,\dots,n}$ by $[n]$, and the set of all size-$k$ subsets of $\setA$ by $\binom{\setA}{k}$. 
 We denote the real-valued Gaussian vector distribution with mean $\vecmu$ and covariance matrix $\matA$ by  $\normal(\vecmu, \matA)$, and the Bernoulli distribution with parameter~$p$ by $\Ber{p}$. We denote convergence in distribution by $\stackrel{D}{\rightarrow}$. Finally, $x^+ = \max\{0,x\}$; 
$\ind{\cdot}$ is the indicator function; $\diag(x_1, \dots, x_n)$ denotes the diagonal matrix with $(x_1, \dots, x_n)$ as the diagonal; $\ln$ is the natural logarithm. 

\section{System Model} \label{sec:model}
We consider a stationary memoryless Gaussian {UMA} scenario, in which $\Ka$ users transmit their
messages to a receiver over $\uN$ channel uses. The users are assumed to be clustered according to their
path loss. Specifically, users within the same cluster~$\ell$ experience the same path
loss~$g_\ell$. For simplicity, we consider the case of $2$ clusters. Our results can however be
readily generalized to an arbitrary finite number of clusters.
We focus on the impact of heterogeneous path loss, and do not model small-scale fading. 
Let $\rvecx_{\ell,k} \in \reals^\uN$ be the signal transmitted by the $k$th user in cluster~$\ell$. The received signal is 
\begin{equation}\label{eq:io-relation}
    \rvecy = g_1 \sum_{k=1}^{\Kaf} \rvecx_{1,k} + g_2 \sum_{k=1}^{\Kas} \rvecx_{2,k} + \rvecz
\end{equation}
where $\Kaf$ and $\Kas$, with $\Kaf + \Kas = \Ka$, are the number of active users that belong to
clusters $1$ and $2$, respectively, and $\rvecz \sim \normal(\veczero, \matidentity_n)$ is the
Gaussian noise, which is independent of the transmitted signals. We consider the power constraint
$\|\rvecx_{\ell,k}\|^2 \le \n\power$ for all $\ell$ and $k$. We also assume that $\Ka, \Kaf$, and
$\Kas$ are fixed and known to the receiver. 
Finally, we assume that each user has perfect knowledge of its path loss.
As already pointed out, this information can be obtained in practical systems via the downlink
control information transmitted by cellular base stations.

For this channel, an $(\M,\n,\power,\epsilon)$ {UMA} code with codebook size $\M$ and codeword length $\n$ consists of 
\begin{itemize}
    \item two encoding functions $f_{\ell} : [\M] \to \reals^{\n}$, $\ell\in\{1,2\}$, that produce the
      transmitted codeword $\rvecx_{\ell,k} = f_\ell(W_k)$, satisfying the power constraint, of
      user~$k$ in cluster~$\ell$, for a user message $W_k$ uniformly distributed over $[\M]$;
    
    \item a decoding function $g: \reals^{\n} \to \binom{[\M]}{\Ka}$ that provides an estimate
      $\widehat{\setW} = \{\widehat W_1, \dots, \widehat W_{|\widehat \setW|}\} = g(\rvecy)$ of the
      list of transmitted messages.
\end{itemize}
The decoding function satisfies the following constraint on the {PUPE}: 
\begin{equation}
     \frac{1}{|\widehat \setW|} \sum_{i = 1}^{|\widehat \setW|} \Pr[\widetilde W_i \notin \widehat \setW] \le \epsilon, \label{eq:def_Pe}
\end{equation}
where $\widetilde \setW = \{\widetilde W_1, \dots, \widetilde W_{|\widetilde \setW|}\}$ denotes the set of distinct elements of $\setW = \{W_1, \dots, W_{\Ka}\}$. In~\eqref{eq:def_Pe}, we use the convention $0/0 = 0$ to circumvent the case $|\widehat \setW| = 0$. 

The main difference between the definition of {UMA} coding scheme just provided and the one
originally provided in~\cite{PolyanskiyISIT2017massive_random_access} is that we allow the encoder
to depend on the cluster. This allows us to model location-based codebook partitioning.

\section{Random Coding Bound}
\label{sec:bound}
In this section, we first derive a random-coding achievability bound on the {PUPE}
achievable over the channel~\eqref{eq:io-relation} that exploits location-based codebook
partitioning.
Then, we provide in Section~\ref{sec:common_codebook} two bounds for the case in which
$f_{1}=f_{2}$, which we refer to as common codebook case.
The first bound relies on joint decoding, whereas the second bound relies on interference
cancellation. 

\subsection{Location-Based Codebook Partitioning} \label{sec:separateCodebooks_jointDec}
We assume that the codewords of each of the two codebooks $\setC_{\ell} = \{\rvecc_{\ell,m}\}_{m=1}^\M$, $\ell \in \{1,2\}$
are drawn independently (across both $\ell$ and $m$) from a $\normal(\veczero, \power' \matidentity_{\n})$ distribution, 
for a fixed $\power' < \power$.
To convey message $W_{k}$, an
active user $k$ in cluster~$\ell$ transmits $\rvecc_{\ell,W_k}$, provided that $\|\rvecc_{\ell,W_k}\|^2
\le \n \power$. Otherwise, the user transmits the all-zero codeword. That is, 
\begin{align}
    f_{\ell}(W_k) = \rvecc_{\ell,W_k} \ind{\|\rvecc_{\ell,W_k}\|^2 \le \n \power}.    
\end{align}
We consider a maximum likelihood decoder, whose output is the set of estimated messages 
$\widehat \setW = \widehat \setW_1 \cup \widehat \setW_2$, where 
\begin{equation} \label{eq:jointDec}
    (\widehat \setW_1, \widehat \setW_2) = \!\!\argmin_{\substack{\setW'_1, \setW'_2 \subset [\M] \colon  \setW'_1 \cap \setW'_2 = \emptyset, \\ |\setW'_1| = \Kaf, |\setW'_2| = \Kas}} \!\!\!\!\!\!\!\!\|\rvecy - g_1 c_1(\setW'_1) - g_2 c_2(\setW'_2)\|
\end{equation}
with $c_\ell(\setW') = \sum_{w\in\setW'} \rvecc_{\ell, w}$, $\ell \in \{1,2\}$.

An error analysis of this random-coding scheme yields the following joint-decoding achievability bound. 

\begin{theorem}[Random-coding bound, location-based codebook partitioning, joint decoding] \label{th:bound}
  Fix $\power' < \power$. For the Gaussian {UMA} channel~\eqref{eq:io-relation}  there exists an $(\M, \n, \power, \epsilon)$ random-access code for which
    \begin{align} \label{eq:bound}
        \epsilon \le \sum_{t= 0}^{\Ka} \frac{t}{\Ka} p_t + p_0
    \end{align}
    where 
    \begin{align}
        p_0 &= 1 - \frac{\M!}{\M^{\Ka}(\M-\Ka)!} + \Ka\frac{\Gamma(\frac{\n}{2},\frac{\n P}{2\power'})}{\Gamma(\n/2)}, \label{eq:p0}
        \end{align}
        \begin{align}
        p_t &= \sum_{t\sub{MD1} \in \setT_t} \sum_{t\sub{FP1} \in \setT_t} \sum_{t\sub{AC1} \in \setT'_{t,t\sub{MD1},t\sub{FP1}}} \notag\\
        &\quad \exp(-\n E(t, t\sub{MD1}, t\sub{FP1}, t\sub{AC1})), \label{eq:pt} \\
        E(t,t\sub{MD1},~&t\sub{FP1}, t\sub{AC1}) \notag \\\quad &= \max_{\rho_1, \rho_2, \rho_3 \in [0,1]} - \rho_1\rho_2\rho_3 (R_1 + R_2) \notag \\
        &\quad - \rho_2 \rho_3 (R_3 + R_4) - \rho_3 (R_5 + R_6) \notag \\
        &\quad  + E_0(\rho_1,\rho_2,\rho_3), \label{eq:Et} \\
        R_1 &= \frac{1}{\n} \ln \binom{\Kaf - t\sub{MD1}}{t\sub{AC1}}, \label{eq:R1} \\
        R_2 &= \frac{1}{\n} \ln \binom{\Kas - t + t\sub{MD1}}{t\sub{AC1} + t\sub{MD1} - t\sub{FP1}}, \label{eq:R2} \\
        R_3 &= \frac{1}{\n} \ln \binom{\M-\Ka}{t\sub{FP1}}, \label{eq:R3} \\
        R_4 &= \frac{1}{\n} \ln \binom{\M-\Ka -t\sub{FP1}}{t-t\sub{FP1}}, \label{eq:R4} \\
        R_5 &= \frac{1}{\n} \ln \binom{\Kaf}{t\sub{MD1}}, \label{eq:R5} \\
        R_6 &= \frac{1}{\n} \ln \binom{\Kas}{t - t\sub{MD1}}, \label{eq:R6} \\
        \!E_0(\rho_1,\rho_2,\rho_3) &=  \max_\lambda  \rho_3 \rho_2 a_1+ \rho_3 a_2 + \frac{1}{2}\ln(1\!-\!2\rho_3b), \label{eq:E0} \\
        a_1&=\frac{\rho_1}{2}
        \ln(1 \!+\! 2\lambda \kappa\sub{AC} \power') \!+\! \frac{1}{2}\ln(1 \!+\! 2\mu_1 \kappa\sub{FP} \power'),\! \label{eq:a_1}\\
        a_2 &=\frac{1}{2}\ln(1+ 2\mu_2 \kappa\sub{MD}\power'), \label{eq:a2}\\
    b &=\lambda \rho_1 \rho_2-\frac{\mu_2}{1+ 2\mu_2 \kappa\sub{MD} \power'}, \label{eq:b}\\
    \mu_1 &= \frac{\lambda \rho_1}{1+ 2\lambda \kappa\sub{AC} \power'},\label{eq:mu1}\\
    \mu_2 & = \frac{\mu_1 \rho_2}{1+2\mu_1 \kappa\sub{FP} \power'},\label{eq:mu2}\\
        \kappa\sub{AC} &= (g_1^2 + g_2^2)(2t\sub{AC1} + t\sub{MD1} \!-\! t\sub{FP1}), \label{eq:kappaAC_joint} \\
        \kappa\sub{MD} &= g_1^2 t\sub{MD1} + g_2^2 (t - t\sub{MD1}), \label{eq:kappaMD} \\
        \kappa\sub{FP} &= g_1^2 t\sub{FP1} + g_2^2 (t - t\sub{FP1}), \label{eq:kappaFP} \\
        \setT_t &= [(t-\Kas)^+ : \min\{t, \Kaf\}], \label{eq:Tt}\\
        \setT'_{t,t\sub{MD1},t\sub{FP1}} &= [\max\{0, t\sub{FP1} - t\sub{MD1}\} : \notag \\
    &\quad~~ \min\{\Kaf - t\sub{MD1}, \Kas - t + t\sub{FP1}\}].\! \label{eq:Ttt}
    \end{align}
\end{theorem}
\begin{proof}
    The proof follows similar steps as the proof of~\cite[Th.~1]{PolyanskiyISIT2017massive_random_access} and~\cite[Th.~1]{Ngo2023_randomActivity} namely, a change of measure and a Gallager-type error exponent analysis that makes use of Chernoff bound, Gallager's $\rho$-trick, and Gaussian statistics. 
    One fundamental difference compared to~\cite{PolyanskiyISIT2017massive_random_access} is that the decoder, when analyzing codewords received at power $\power g^2_2$, may put out a set of messages, which we denote by $\setW\sub{AC1}$, that are false positives from the perspective of cluster $2$, but happen to coincide with
    the messages from cluster~$1$, and vice versa. 
    We refer to these messages as ``accidentally correct''~(AC) messages. 
    Compared to~\cite{PolyanskiyISIT2017massive_random_access}, this results in an additional union bound over $t\sub{AC1} = |\setW\sub{AC1}|$ (see~\eqref{eq:pt}). 
    See \arxivonly{Appendix~\ref{proof:bound}} \confonly{\cite[App.~B]{Krishnan2026location}} for details.
\end{proof}

\subsection{Common Codebook}\label{sec:common_codebook}

We now consider the common codebook case $\setC_1 = \setC_2 = \{\rvecc_1, \dots, \rvecc_\M\}$. 

\subsubsection{Joint Decoding}
We have the following result. 

\begin{cor}[Random-coding bound, common codebook, joint decoding] \label{cor:bound_separate}
  Fix $\power' < \power$. For the Gaussian {UMA} channel~\eqref{eq:io-relation}, 
    there exists an $(\M, \n, \power,  \epsilon)$ random-access code for which $\epsilon$ is bounded as in~\eqref{eq:bound} with~\eqref{eq:kappaAC_joint} replaced by
    \begin{equation}
        \kappa\sub{AC} = (g_1- g_2)^2 (2t\sub{AC1} + t\sub{MD1} - t\sub{FP1}). \label{eq:kappaAC_separate}
    \end{equation}
\end{cor}
\begin{proof}
    The proof follows along the same lines as the proof of Theorem~\ref{th:bound}, with the fundamental difference that the assumption of common codebook causes an increase in $\kappa\sub{AC}$, which represents the variance of a term related to the accidentally corrected messages. See \arxivonly{Appendix~\ref{proof:bound_separate}}\confonly{\cite[App.~C]{Krishnan2026location}}.
\end{proof}



    

\subsubsection{Interference-Cancellation Decoding} 
Assume without loss of generality that $g_1 \ge g_2$. We now consider an interference-cancellation decoder that operates by first decoding messages from cluster~$1$ as
\begin{equation}
    \widehat \setW_1 = \argmin_{\setW'_1 \subset [M] \colon |\setW_1'| = \Kaf} \|\rvecy - g_1 c(\setW'_1)\| \label{eq:SIC_dec_1}
\end{equation}
with $c(\setW) = \sum_{w\in\setW} \rvecc_{w}$,
and then canceling the interference from  cluster~1 to decode messages coming from cluster~2 as
\begin{equation}
    \widehat \setW_2 = \argmin_{\setW'_2 \subset [M] \setminus \widehat\setW_1 \colon |\setW_2'| = \Kas} \|\rvecy - g_1 c(\widehat \setW_1) - g_2 c(\setW'_2)\|. \label{eq:SIC_dec_2}
\end{equation}
We state a random-coding bound for this strategy in the following theorem.
\begin{theorem}[Random-coding bound, common codebook, interference-cancellation decoding] \label{th:bound_IC}
  Fix $\power' < \power$. For the Gaussian {UMA} channel~\eqref{eq:io-relation}, there exists an $(\M, \n, \power,  \epsilon)$ random-access code for which $\epsilon$ is bounded as in~\eqref{eq:bound} with 
    $E_0(\rho_1,\rho_2,\rho_3)$ replaced by 
    \vspace{-.2cm}
    \begin{align}
        E_0(&\rho_1,\rho_2,\rho_3) =  \max_{\lambda_1, \lambda_2} \Big(\frac12 \rho_1\rho_2\rho_3 \ln\det(\matidentity_2 - 2 \matSigma\sub{AC} \matU\sub{AC}) \notag \\
        &\qquad + \frac12 \rho_2 \rho_3 \ln\det(\matidentity_2 - 2 \matSigma\sub{FP} \matU\sub{FP}) \notag \\
        &\qquad + \frac12 \rho_3 \ln\det(\matidentity_2 - 2 \matSigma\sub{MD} \matU\sub{MD}) \notag \\
        &\qquad + \frac12 \ln\det(\matidentity_2 - 2 \matSigma\sub{ZC} \matU\sub{ZC})\Big), \\ 
        \matSigma\sub{AC} &= \diag(t\sub{AC1} \power', (t\sub{AC1} + t\sub{MD1} - t\sub{FP1}) \power'), \label{eq:Sigma_AC} \\ 
        \matSigma\sub{FP} &= \diag(t\sub{FP1} \power', (t - t\sub{FP1}) \power'), \label{eq:Sigma_FP}\\ 
        \matSigma\sub{MD} &= \diag(t\sub{MD1} \power', (t - t\sub{MD1}) \power'), \label{eq:Sigma_MD}\\ 
        \matSigma\sub{ZC} &= \diag(1, (\Kas - t + t\sub{FP1} - t\sub{AC1}) \power'), \label{eq:Sigma_ZC}\\ 
        \matU\sub{AC} &= \begin{bmatrix}
        u\sub{AC}\of{1} & u\sub{AC}\of{2} \\
        u\sub{AC}\of{2} & u\sub{AC}\of{3}
        \end{bmatrix}, \label{eq:U_AC} \\
        u\sub{AC}\of{1} &= -\lambda_1 g_1^2 -\lambda_2 g_2(g_2 - 2 g_1), \\
        u\sub{AC}\of{2} &= \lambda_1 g_1 (g_1 - g_2) + \lambda_2 g_2 (g_2 - 2 g_1), \\
        u\sub{AC}\of{3} &= -\lambda_1 g_1 (g_1 - 2g_2) - \lambda_2 g_2 (g_2 - 2g_1), \\ 
        \matU\sub{FP} &= 2 \rho_1  \matLambda\sub{FP} \matQ\sub{FP} - \rho_1 \!\begin{bmatrix}
            \lambda_1 g_1^2 & \!\!\lambda_2g_1 g_2 \\ \lambda_2g_1 g_2\!\! & \lambda_2 g_2^2 
        \end{bmatrix}, \label{eq:U_FP} \\ 
        \matQ\sub{FP} &= \begin{bmatrix}
            \lambda_1 g_1^2 - \lambda_2 g_1 g_2 & \lambda_2 g_2(g_1 - g_2) \\
            -\lambda_1 g_1 (g_1 \!-\! g_2) + \lambda_2 g_1 g_2\! & \!-\lambda_2 g_2 (g_1 \!-\! g_2)
        \end{bmatrix}\!, \\ 
        \matLambda\sub{FP} &= \tp{\matQ}\sub{FP} (\matSigma\sub{AC}^{-1} \!-\! 2 \matU\sub{AC})^{-1}, \\ 
        \matU\sub{MD} &= 2\rho_2 \matLambda\sub{MD} \matQ\sub{MD} + 2 \rho_2 \rho_1 \matLambda_{\matM} \matM \notag \\
    &\quad - \rho_2 \rho_1 \begin{bmatrix}
        \lambda_1 g_1^2 & (\lambda_1 + \lambda_2) g_1 g_2 \\
        (\lambda_1 + \lambda_2) g_1 g_2 & \lambda_2 g_2^2
    \end{bmatrix}, \label{eq:U_MD} \\ 
    \matM &= \lambda_1 g_1 g_2 \begin{bmatrix}
        0 & - 1 \\ 0 & 1
    \end{bmatrix}  - \matQ\sub{FP}, 
    \\ 
    \matQ\sub{MD} &= 2\rho_1 \matLambda\sub{FP} \matM + \rho_1 \begin{bmatrix}
        \lambda_1 g_1^2 & (\lambda_1 + \lambda_2) g_1 g_2 \\
        \lambda_2 g_1 g_2 & \lambda_2 g_2^2
    \end{bmatrix}, \\ 
    \matLambda\sub{MD} &=
    \tp{\matQ}\sub{MD} (\matSigma\sub{FP}^{-1} - 2\matU\sub{FP})^{-1}, \\ 
    \matLambda_{\matM} &= \tp{\matM} (\matSigma\sub{AC}^{-1} - 2\matU\sub{AC})^{-1}, \\ 
    \matU\sub{ZC} &= 2\rho_3 \tp{\matQ}\sub{ZC}(\matSigma\sub{MD}^{-1} - 2  \matU\sub{MD})^{-1} \matQ\sub{ZC} \notag \\
    &\quad + 2 \rho_3 \rho_2 \tp{\matLambda}\sub{ZCa} (\matSigma\sub{FP}^{-1} - 2\matU\sub{FP})^{-1} \matLambda\sub{ZCa} \notag \\
    &\quad + 4 \rho_3 \rho_2 \rho_1 \tp{\matLambda}\sub{ZCb} (\matSigma\sub{AC}^{-1} - 2\matU\sub{AC})^{-1} \matLambda\sub{ZCb}, \label{eq:U_ZC} \\ 
    \matQ\sub{ZC} &= \rho_2 \rho_1  \!\bigg(\!\!-\matidentity_2 \!+\! (4\matLambda\sub{MD} \matLambda\sub{FP} \!+\! 2 \matLambda_{\matM}) \!\! \begin{bmatrix}
        -1\!\! & 1 \\ 1 & \!\!- 1
    \end{bmatrix} \!+\! 2 \matLambda\sub{MD} \!\bigg) \notag \\&\quad \cdot \begin{bmatrix}
        \lambda_1 g_1 & \lambda_1 g_1 g_2 \\
        \lambda_2 g_2 & 0
    \end{bmatrix}\!, \\ 
    \matLambda\sub{ZCa} &= \rho_1 \bigg(\matidentity_2 + 2 \matLambda\sub{FP} \begin{bmatrix}
        -1 & 1 \\ 1 & - 1
    \end{bmatrix}\bigg)  \begin{bmatrix}
        \lambda_1 g_1 & \!\lambda_1 g_1 g_2 \\
        \lambda_2 g_2 & 0
    \end{bmatrix}\!, \\
    \matLambda\sub{ZCb} &= \begin{bmatrix}
        -1 & 1 \\ 1 & - 1
    \end{bmatrix}  \begin{bmatrix}
        \lambda_1 g_1 & \lambda_1 g_1 g_2 \\
        \lambda_2 g_2 & 0
    \end{bmatrix}.
    \end{align}
\end{theorem}
\begin{proof}
    The proof follows similar steps as the proof of Theorem~\ref{th:bound}, with some extensions to capture the sequential nature of the decoding operations~\eqref{eq:SIC_dec_1} and~\eqref{eq:SIC_dec_2}. 
    See \arxivonly{Appendix~\ref{proof:bound_IC}}\confonly{\cite[App.~D]{Krishnan2026location}}.
\end{proof}

\section{Replica Method Prediction} \label{sec:replica}
\par To obtain insights on system performance, we provide next a large-system limit characterization of the {PUPE}  achievable over the channel \eqref{eq:io-relation}. 
To derive the result, we rely on the replica method, in line with its
application to multiuser detection systems~\cite{tanaka2002statistical,Guo2005,guo2005performance}, and the
original {UMA} problem~\cite{polyanskiy2020slides}.  

%
 \par To state the main results of this section, we first introduce the following definition of multiuser efficiency.
 \begin{definition}
     Let 
     \begin{equation}
         B=\sqrt{a} A+N,
         \label{eq:decoup_channel}
     \end{equation}
     with $a>0$, 
     $A\distas\Ber{p}$, $p\in(0,1)$, and   $N\distas\mathcal{N}(0,1)$.
     Let $I(a)$ denote the mutual information $I(A;B)$ and let $G$ be a binary random variable taking values in $\{g_1,g_2\}$ with probability $\mathbb{P}[G=g_\ell]=\gamma_\ell$. Then, for every $\beta>0$, the multiuser efficiency $\eta$ is defined as 
     \begin{equation}
     \eta = \arg\min_{x}\lefto(\beta \Exop_G[I(x\power G^2)] + \tfrac{1}{2}(x - 1 - \ln x)\right).
         \label{eq:eta}
     \end{equation}
 \end{definition}
\vspace{-.1cm}
\subsection{Location-Based Codebook Partitioning}
   
We describe the random codebook corresponding to cluster~$\ell$ as a matrix $\msC_\ell$ with $\M$ columns drawn independently from a  
 $\normal(\veczero, {\matidentity_\n}/{\n} )$
 distribution.  For a fixed ${\mu \in (0,1)}$, we consider the regime in which the total number of active users satisfies $\Ka = \mu\n$.  In contrast to Section \ref{sec:bound}, where each user independently selects a message uniformly at random yielding a message selection probability $1-(1-1/\M)^{\Ka} \approx \Ka/\M$, we consider a setting in which, within each cluster $\ell$ with $\Ka_\ell$ active users, each message is selected independently according to a $\Ber{\Ka_\ell/\M}$
 distribution. 
 Let $\alpha_\ell=\uK_{\rm a \ell}/\Ka$ and $\beta=2\M / \n $. 
 Furthermore, let $\rvecu_\ell'$ denote an $\M \times 1$  vector with entries drawn independently from a $\Ber{2\alpha_\ell\mu/\beta}$ distribution,
 representing message selection in cluster $\ell$. Let finally $\rvecz'\distas\mathcal{N}(\zv,\matidentity_\n)$.  We analyze the large-system performance achievable over the channel 
\begin{equation}
\rvY'=\sqrt{\power} g_1 \msC_1\rvecu_1' + \sqrt{\power} g_2 \msC_2\rvecu_2'+\rvecz'.
    \label{eqn:replica_sub}
\end{equation}

We have the following result.
\vspace{-.1cm}
\begin{claim}[Replica decoupling, location-based codebook partitioning]\label{thm:uma_pl_dep_codebook_replica}
Let $\ell\in\{1,2\}$. Fix $p_\ell=2\alpha_\ell\mu/\beta$, and $\gamma_\ell=1/2$.  
Denote by
\begin{equation}
V_{\ell,\n}' = \Pr\lefto[[\rvecu_{\ell}']_1 = 1 \mid \rvY', \msC_{1}, \msC_{2}\right] 
\label{eqn:vn_defn}
\end{equation}
the marginal posterior probability associated with~\eqref{eqn:replica_sub}.
Consider also the following scalar channels for $\ell \in\{1,2\}$:
\begin{equation}
     B_\ell=\sqrt{\power \eta}g_\ell A_\ell+N_\ell.
     \label{eq:replica_decoup_channels_loc_based}
\end{equation}
 Here, $\eta$ is as in \eqref{eq:eta}, $A_\ell\distas\Ber{p_\ell\gamma_\ell}$,  and $N_\ell\distas\mathcal{N}(0,1)$, independent of $A_\ell$. Let $\M, \n, \Ka \to \infty$ with $\beta$, $\mu$, and $\alpha_\ell$ held constant. 
Then, 
\begin{equation}\label{eq:marginal-posterior-partitioning}
    V_{\ell,\n}' \stackrel{D}{\longrightarrow} \Pr[A_\ell = 1 \mid B_\ell].
\end{equation}

%
\end{claim}
\begin{proof}
  The proof follows from~\cite[Prop.~1]{guo2005performance}. See \confonly{\cite[App.~E]{Krishnan2026location}}.
\end{proof}
\vspace{-.3cm}
\begin{remark}
\label{rem:replica_loc_based}
    Claim \ref{thm:uma_pl_dep_codebook_replica} implies that, in the large-system limit, the channel~\eqref{eqn:replica_sub} decouples into the two scalar channels in~\eqref{eq:replica_decoup_channels_loc_based}.
    Note in particular that the two scalar channels have deterministic path losses known to the receiver.
    Following an approach similar to the one detailed in~\cite{polyanskiy2020slides}, we can use these two scalar channels to obtain a large-system characterization of the {PUPE}.
    Specifically, let $\epsilon_\ell$ be the solution of 
    \begin{equation}
        \sqrt{\power\eta} g_\ell =Q^{-1}(\epsilon_\ell)+Q^{-1}(p_\ell\gamma_\ell\epsilon_\ell/(1-p_\ell\gamma_\ell)).
    \end{equation}
    In the large-system limit, the {PUPE} is given by  
    $(\epsilon_1+\epsilon_2)/2$.
\end{remark}
\subsection{Common Codebook}
\label{sec:replica-common}
We denote the common codebook  as a matrix $\msC$ with $\M$ columns drawn independently from a $\normal(\veczero, \matidentity_\n/\n )$ distribution. We let $\mu$ and $\alpha_\ell$  be defined as before, but set now $\beta=\M/\n$, and use a $\Ber{\Ka/\M}$  distribution for message selection. Note that $\Ka/\M=\mu/\beta$. Let $\rvecu''$ be the binary vector describing the message selection, and $G$ a random variable taking values in $\{g_1,~g_2\}$ with $\mathbb{P}[G=g_\ell]=\alpha_\ell$.  
  Finally, let $\rvecz''\distas\mathcal{N}(\zv,\matidentity_\n)$. We analyze the performance achievable over the channel 
\begin{equation}
\rvY''= \sqrt{P}G \msC \rvecu''+\rvecz''.
    \label{eqn:replica_com}
\end{equation}
We have the following result. 
\begin{claim}[Replica decoupling, common codebook]
Let \mbox{$\ell\in\{1,2\}$}. Fix $p=\mu/\beta$, and $\gamma_\ell=\alpha_\ell$. Denote by  
\begin{equation}
V_{\n}'' = \Pr\lefto[[\rvecu'']_1 = 1 \mid \rvY'', \msC''\right], 
\label{eqn:vn2_defn}
\end{equation}
the marginal posterior probability associated with \eqref{eqn:replica_com}. Consider the scalar channel
\begin{equation}
    B''=\sqrt{P\eta}GA''+N''. 
\end{equation}
Here, $\eta$ is as in \eqref{eq:eta}, $A''\sim\Ber{p}$, $G$ takes values in $\{g_1,g_2\}$ with \mbox{$\mathbb{P}[G=g_\ell]=\alpha_\ell$}, and $N''\sim \mathcal{N}(0,1)$. Furthermore, these three random variables are mutually independent. Let $\M,\n,\Ka  \to \infty$ with $\beta$, $\mu$, and $\alpha_\ell$ held constant.
Then,
\begin{equation}
    V_{\n}'' \stackrel{D}{\longrightarrow} \Pr[A'' = 1 \mid B''].
\end{equation}
\label{thm:uma_pl_indep_codebook_replica}
\end{claim}
\begin{proof}
  The proof follows from~\cite[Claim~1]
  {Guo2005}. See \confonly{\cite[App.~F]{Krishnan2026location}}.
\end{proof}
\begin{remark}
\label{rem:replica_common}
Note that, in contrast to the previous case, the equivalent scalar channel in the common codebook case is a single fading channel, with fading coefficient $G$ not known to the receiver.
The large-system characterization of the {PUPE} for this scenario can again be carried out along the lines of~\cite{polyanskiy2020slides}.
However, the presence of $G$ precludes a closed-form expression.
Specifically, in the large-system limit, the {PUPE} is given by 
\begin{equation}
        \epsilon=\frac{1-p}{p}\mathbb{P}_{0}\lefto[\ln \frac{\dv{\mathbb{P}_{1}}}{\dv{\mathbb{P}_{0}}} \ge \theta\right]
    \end{equation}
    where $\theta$ is determined by imposing that 
    \begin{equation}
      \mathbb{P}_{1}\lefto[\ln \frac{\dv\mathbb{P}_{1}}{\dv {\mathbb{P}_{0}}}  \ge \theta\right] + \frac{1-p}{p}\mathbb{P}_{0}\lefto[\ln \frac{\dv\mathbb{P}_{1}}{\dv{\mathbb{P}_{0}}} \ge \theta\right] = 1. 
    \end{equation}
   Here, $\mathbb{P}_0=\mathcal{N}(0,1)$ and $\mathbb{P}_1=\sum_{\ell\in\{1,2\}}\gamma_\ell\mathcal{N}(\sqrt{\power \eta}g_\ell,1)$. 
\end{remark}
\section{Numerical Results} \label{sec:results}
We evaluate the minimum energy per bit, defined as $\n \power/(2\B)$, required to achieve a {PUPE} of $0.01$ as a function of the number of active users $\Ka$ for the case in which each user maps messages of $\B=128$ bits to codewords of length $\n=\num{30000}$.
Throughout, we set $g_1=1$, $g_2=0.8$, and $\Kas=2\Kaf$.
In Fig.~\ref{fig:numerical_results}, we depict the
random-coding bound with location-based codebook partitioning in Theorem~\ref{th:bound} as well as those corresponding to common codebook with joint decoding (Corollary~\ref{cor:bound_separate}) and common codebook with  interference-cancellation decoding (Theorem~\ref{th:bound_IC}).
In the figure, we also depict the replica method predictions, obtained via the expressions given in Remark~\ref{rem:replica_loc_based} and Remark~\ref{rem:replica_common}.
We also use the normalization with respect to the effective number of bits suggested in~\cite{polyanskiy2020slides}.
Finally, we provide simulation results for a {CCS} scheme~\cite{Amalladinne2020} in which the decoder operates according to the multisource-{AMP} framework proposed  in~\cite{Cakmak2025}. 
Specifically, in the resulting scheme, each message is divided into $16$ blocks of $16$ bits each. The {CCS} inner decoder performs multisource-{AMP} signal reconstruction,
while the outer decoder stitches together the reconstructed
signals, to ensure that they form valid codewords.

As shown in Figure \ref{fig:numerical_results}, for the channel model considered in the paper, location-based codebook partitioning results in a consistent reduction of the minimum energy  per bit (although this reduction is marginal when $\Ka$ is small) across all types of curves depicted in the figure (random-coding bounds, replica-method predictions, performance of {CCS} schemes).
Also, for the parameter values considered in this section, interference cancellation exhibits poor performance ($\SI{6}{\dB}$ above the location-based codebook partitioning bound for $\Ka=12$).

\begin{figure}[t!]
    \centering
    \begin{tikzpicture}[
spy using outlines= {circle, magnification=3.5, connect spies}
]
    \tikzstyle{every node}=[font=\footnotesize] 
    \begin{axis}[
    scale only axis,
    width=3.0in,
    height=1.6in,
    grid=both,
    grid style={line width=.1pt, draw=gray!10},
    major grid style={line width=.2pt,draw=gray!50},
    xmin=0,
    xmax=300,
    xlabel={\scriptsize number of active users, $\Ka$},
    ymin=0,
    ymax=8,
    ytick = {0, 1, 2, 3, 4, 5, 6, 7, 8},
    ylabel={\scriptsize minimum energy per bit (dB)},
    ylabel style = {xshift=-1mm},
    legend style={draw=none,opacity=.9,at={(0.96,0.57)},anchor=south east,legend cell align=left, 
    nodes={font=\scriptsize}
    },
    ]

    \addplot [color=black, mark=*, line width=1.2pt, mark size=1pt]
    table[row sep=crcr]{
        3       6.41    \\
        12      6.46    \\
        42      6.53    \\ 
        72      6.57    \\ 
        102     6.6     \\ 
        132     6.62    \\ 
        162     6.66    \\ 
        192     6.83     \\    
        222     6.85    \\
        252     6.87    \\
        282     6.96    \\
        312     7.02    \\
    };
    \addlegendentry{common codebook};

    \addplot [color=red, mark=*, dashed, line width=1.2pt, mark size=1pt]
    table[row sep=crcr]{%
        3       6.26     \\
        12      6.31    \\
        42      6.4     \\ 
        72      6.47    \\ 
        102     6.51    \\ 
        132     6.52    \\ 
        162     6.53    \\ 
        192     6.59    \\ 
        222     6.65    \\
        252     6.72    \\ 
        282     6.75    \\
        312     6.80    \\
    };
    \addlegendentry{location-based codebook partitioning};

    \addplot [color=black, mark=*, line width=1.2pt, mark size=1pt]
    table[row sep=crcr]{%
        3       2.3079      \\
        15      2.4873      \\ 
        30      2.6145      \\ 
        60      2.8018      \\ 
        90      2.9484      \\ 
        120     3.0730      \\ 
        150     3.1860      \\    
        180     3.3011      \\
        240     3.44        \\
        300     3.88        \\
        330     4.2828      \\
    };

    \addplot [color=red, mark=*, dashed, line width=1.2pt, mark size=1pt]
    table[row sep=crcr]{%
        3       2.23529114546391        \\    
        15      2.27156303353108        \\ 
        30      2.3009539657918         \\ 
        60      2.34100833280823        \\ 
        90      2.37174797803295        \\ 
        120     2.40591896864964        \\ 
        150     2.4164      \\    
        180     2.4363      \\
        240     2.7137      \\
        300     3.0805      \\
    };
    
    \addplot [color=black, mark=*, line width=1.2pt, mark size=1pt]
    table[row sep=crcr]{%
        3       1.6434      \\
        15      1.5613      \\ 
        60      1.5256      \\ 
        120     1.4957      \\ 
        180     2.1190      \\    
        240     2.83373     \\
        300     3.2113      \\    
    };

    \addplot [color=black, mark=*, line width=1.2pt, mark size=1pt]
    table[row sep=crcr]{%
        3       5.1         \\
        9       7.8075      \\
        12      8.1         \\
    };

    \addplot [color=red, mark=*, dashed, line width=1.2pt, mark size=1pt]
    table[row sep=crcr]{%
        3       1.5794      \\
        15      1.4806      \\ 
        60      1.4443      \\ 
        120     1.4684      \\ 
        180     1.4564      \\    
        240     1.4504     \\
        300     1.9814      \\   
    };

    \coordinate (Center) at (axis cs:192, 6.72);
    \coordinate (Center3) at (axis cs:60, 1.5);

    \node[align = center] at (axis cs:123, 7.4) () {\scriptsize CCS with multisource AMP};
    \node[red, align = center] at (axis cs:108, 1.9) () {\scriptsize Theorem~1};
    \node[align = center] at (axis cs:163, 2.72) () {\scriptsize Corollary~1};
    \node[align = center] at (axis cs:120, 0.75) () {\scriptsize replica method prediction};
    \node[align = center] at (axis cs:37.5, 5.28) () {\scriptsize Theorem~2};

    \node [black, inner sep=0, rotate=45] at (axis cs:185,7.12) {\small $\uparrow$};
    \node [red, inner sep=0, rotate=235] at (axis cs:80, 2.18) {\small $\uparrow$};
    \node [black, inner sep=0, rotate=235] at (axis cs:135, 2.90) {\small $\uparrow$};
    \node [black, inner sep=0, rotate=235] at (axis cs:69, 1.15) {\small $\uparrow$};
    \node [black, inner sep=0, rotate=235] at (axis cs:10, 5.53) {\small $\uparrow$};

    \coordinate (spypoint)     at (252,6.8);
    \coordinate (magnifyglass) at (252,5);

  \end{axis}

    \draw (Center) ellipse [x radius=0.15, y radius=0.2];
    \draw (Center3) ellipse [x radius=0.15, y radius=0.2];
    in node[fill=white] at (magnifyglass);

\end{tikzpicture}    
    \vspace{-1.1cm}
    \caption{Minimum energy per bit to achieve a PUPE of $0.01$ vs.  $\Ka$. } 
    \label{fig:numerical_results}
    \vspace{-4mm}
\end{figure}
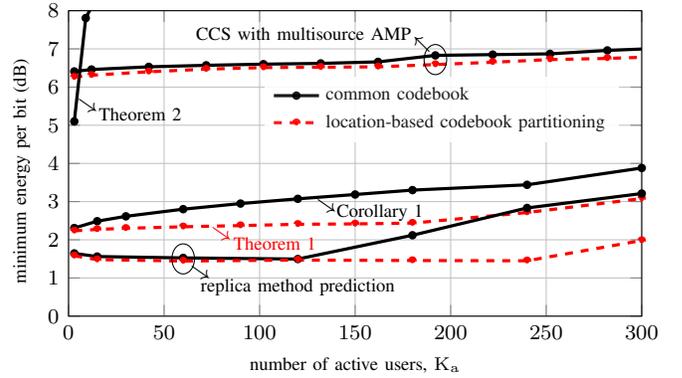


\section{Conclusion} \label{sec:conclusion}

For a Gaussian {UMA} channel characterized by  heterogeneous
path loss, we showed that location-based codebook partitioning outperforms, in terms of the minimum energy per bit required to meet a specific {PUPE}, the conventional {UMA} framework, which utilizes a common codebook for all users. 
These gains were validated through finite-blocklength random coding bounds, replica method large-system limit predictions, and empirical performance of a coding scheme based on {CCS} and multisource-{AMP}.

We anticipate that the energy efficiency gains from location-based codebook partitioning will be even more pronounced in wireless network architectures
featuring distributed access points. Indeed, in such systems, the access points are located so as to ensure uniform quality of service, which should amplify the benefit of the location-based codebook partitioning, as illustrated in~\cite{Cakmak2025} for the case of multisource-{AMP} decoders. Such extension, as well as the inclusion of small-scale fading will be considered in future work.

\clearpage
\bibliographystyle{IEEEtran}  
\bibliography{references}

\arxivonly{ 
\clearpage
\begin{appendices}


\section{Mathematical Preliminaries} \label{app:math_pre}

    The following results will be used in the proof.
    
\begin{lemma}[{Change of measure~\cite[Lemma~4]{Ohnishi2021}}] \label{lem:change_measure}
        Let $p$ and $q$ be two probability measures. Consider a random variable $X$ supported on $\setH$ and a function $f \colon \setH \to [0,1]$. It holds that 
        \begin{align}
            \Exop_p[f(X)] \le \Exop_q[f(X)] + d_{\rm TV}(p,q)
        \end{align}
        where $d_{\rm TV}(p,q)$ denotes the total variation distance between $p$ and $q$.
    \end{lemma}

    \begin{lemma}[{Chernoff bound~\cite[Th. 6.2.7]{DeGroot2012ProbStats}}] \label{lem:Chernoff}
        For a random variable $X$ with moment-generating function $\Exop[e^{t X}]$ defined for all $|t| \le b$, it holds for all $\lambda \in [0,b]$ that
        \begin{align}
            \Prob[X \le x] \le e^{\lambda x} \Exop[e^{-\lambda X}].
        \end{align}
    \end{lemma}

    \begin{cor}[Chernoff joint bound] \label{cor:Chernoff_joint}
        For two random variables $X_1$ and $X_2$ with joint moment-generating function $\Exop[e^{t_1 X_1 + t_2 X_2}]$ defined for all $|t_1| \le b_1$ and $|t_2| \le b_2$, it holds for all $\lambda_1 \in [0,b_1]$, $\lambda_2 \in [0,b_2]$ that
        \begin{align}
            \Prob[X_1 \le x_1, X_2 \le x_2] \le e^{\lambda_1 x_1 + \lambda_2 x_2} \Exop[e^{-\lambda_1 X_1 - \lambda_2 X_2}].
        \end{align}
    \end{cor}
    \begin{proof}
        The proof follows straightforwardly by noting that $X_1 \le x_1, X_2 \le x_2$ implies $\lambda_1 X_1 + \lambda_2 X_2 \le \lambda_1 x_1 + \lambda_2 x_2$, and by apply Lemma~\ref{lem:Chernoff} to the random variable $\lambda_1 X_1 + \lambda_2 X_2$. 
    \end{proof}
    
    \begin{lemma} [{Gallager's $\rho$-trick~\cite[p.~136]{Gallager1968information}}] \label{lem:Gallager}
        It holds that $\Prob[\cup_i A_i] \le (\sum_{i} \Prob[A_i])^\rho$ for every $\rho \in [0,1]$.
    \end{lemma}
    
    \begin{lemma}[Moment-generating function of quadratic forms of a Gaussian vector] \label{lem:chi2}
        Let $\rvecx \in \reals^\n$ and $\rvecx \sim \normal(\vecmu,\matSigma)$. Let $\matA \in \reals^{\n \times \n}$ be a symmetric matrix and $\vecb \in \reals^\n$. For every $\gamma \in \reals$ such that $\matSigma^{-1} + 2\gamma \matA \succ \matzero$, it holds that
        \begin{align}
            &\Exop[\exp(-\gamma(\tp{\rvecx} \matA \rvecx + 2 \tp{\vecb} \rvecx))] \notag \\
            &= \det(\matidentity_N + 2 \gamma \matSigma \matA)^{-1/2} \notag \\
            &\quad \cdot\exp\big(2 \gamma^2\tp{(\matA \vecmu + \vecb)}(\matSigma^{-1} + 2 \gamma  \matA)^{-1}  (\matA \vecmu + \vecb) \notag \\
            &\qquad \quad - \gamma(\tp{\vecmu} \matA \vecmu + 2 \tp{\vecb} \vecmu)\big). \label{eq:quad_Gaussian}
        \end{align}
        In particular, for $\gamma = -1$ and $\vecmu = \veczero$, it holds that
        \begin{align}
            &\Exop[\exp(\tp{\rvecx} \matA \rvecx + 2 \tp{\vecb} \rvecx)] \notag \\
            &= \det(\matidentity_N - 2 \matSigma \matA)^{-1/2} \exp\big(2\tp{ \vecb}(\matSigma^{-1} - 2  \matA)^{-1} \vecb), \label{eq:quad_Gaussian_zeroMean}
        \end{align}
        given that $\matSigma^{-1} \succ 2  \matA$; for $\matSigma = \sigma^2 \matidentity_\n$, $\matA = \matidentity_\n$, and $\vecb = \veczero$, it holds that
        \begin{align}
            \Exop[e^{-\gamma \|\rvecx\|^2}] = (1+2\gamma\sigma^2)^{-\n/2} \exp\bigg(-\frac{\gamma\|\vecmu\|^2}{1+2\gamma\sigma^2}\bigg), \label{eq:mgf_chi2}
        \end{align} 
        for every $\gamma > -\frac{1}{2\sigma^2}$.
    \end{lemma}
    \begin{proof}
        Denote the quadratic form as $Q = \tp{\rvecx} \matA \rvecx + 2 \tp{\vecb} \rvecx$.
        Using the density of $\rvecx\sim\normal(\vecmu,\matSigma)$, 
        we compute $\Exop[e^{-\gamma Q}]$ as
        \begin{align}
        &\Exop[e^{-\gamma Q}]
        =\frac{1}{(2\pi)^{\n/2}\det(\matSigma)^{1/2}}
        \notag \\
        &\cdot \int_{\mathbb{R}^n}
        \exp\!\big(\!
        -\tfrac12\tp{(\vecx\!-\!\vecmu)} \matSigma^{-1}(\vecx\!-\!\vecmu)
        -\gamma \tp{\vecx} \matA\vecx
        -2\gamma \tp{\vecb} \vecx
        \big)\dv\vecx. \label{eq:tmp570} 
        \end{align}
        We define
        $\matK = \matSigma^{-1}+2\gamma \matA$ and $\vech = \matSigma^{-1}\vecmu - 2\gamma \vecb$.
        Since $\matSigma^{-1}$ and $\matA$ are symmetric, $\matK$ is symmetric. The exponent on the right-hand side of~\eqref{eq:tmp570} becomes $-\tfrac12 \tp{\vecx} \matK \vecx + \tp{\vech} \vecx
        -\tfrac12 \tp{\vecmu} \matSigma^{-1}\vecmu$.
        Hence,
        \begin{align}
        \mathbb{E}\!\left[e^{-\gamma Q}\right]
        &=\frac{\exp\!\big(-\tfrac12\tp{\vecmu} \matSigma^{-1}\vecmu\big)}{(2\pi)^{n/2}\det(\matSigma)^{1/2}} \notag \\
        &\quad \cdot 
        \int_{\mathbb{R}^n} \exp\!\Big(-\dfrac12 \tp{\vecx} \matK \vecx + \tp{\vech} \vecx\Big)\dv\vecx. \label{eq:tmp582}
        \end{align}
        
        Under the condition that $\matK$ is positive definite, the standard Gaussian integral identity yields
        \begin{multline}
        \int_{\mathbb{R}^n}\exp\!\Big(-\dfrac12 \tp{\vecx} \matK \vecx + \tp{\vech} \vecx\Big)\dv\vecx
        \\ =(2\pi)^{n/2}\det(\matK)^{-1/2}\exp\!\Big(\dfrac12 \tp{\vech} \matK^{-1}\vech\Big).
        \end{multline}
        Substituting this into~\eqref{eq:tmp582} gives
        \begin{align}
        \Exop\left[e^{-\gamma Q}\right]
        &= \det(\matSigma)^{-1/2}\det(\matK)^{-1/2}
        \notag \\
        &\quad \cdot \exp\!\Big(-\dfrac12\tp{\vecmu} \matSigma^{-1}\vecmu + \dfrac12 \tp{\vech} \matK^{-1}\vech\Big),
        \end{align}
        which leads to~\eqref{eq:quad_Gaussian} after some simplifications of the determinant and exponent terms. The particular cases~\eqref{eq:quad_Gaussian_zeroMean} and~\eqref{eq:mgf_chi2} follow straightforwardly from~\eqref{eq:quad_Gaussian}.
    \end{proof}
\section{Proof of Theorem~\ref{th:bound}}
\label{proof:bound}
		
We analyze the PUPE achieved with the random-coding scheme introduced in Section~\ref{sec:separateCodebooks_jointDec}, averaged over the Gaussian code ensemble. We denote by $\setW_\ell$ the set of messages transmitted by users in cluster~$\ell$, $\ell \in \{1,2\}$, and define $(\widehat \setW_1, \widehat \setW_2)$ as in~\eqref{eq:jointDec}. Furthermore, we denote by $\setW\sub{MD}$ the set of misdetected messages, i.e., $\setW\sub{MD} = \widetilde\setW \setminus \widehat\setW$, and by $\setW\sub{FP}$ the set of false-positive messages, i.e., $\setW\sub{FP} = \widehat\setW \setminus \widetilde\setW$. The PUPE can be expressed as
\begin{equation}
    P\sub{e} = \Exop\lefto[\frac{|\setW\sub{MD}|}{|\widetilde \setW|}\right]. 
    \label{eq:PUPE_frac}
\end{equation}




\subsection{Change of Measure}
We apply Lemma~\ref{lem:change_measure} to the random variable $\frac{|\setW\sub{MD}|}{|\widetilde \setW|}$ to replace the measurement under which the expectation is taken by the one under which: i) the active users transmit distinct messages, i.e., $|\widetilde \setW| = \Ka$ and $\widetilde W_1, \dots, \widetilde W_{\Ka}$ are sampled uniformly without replacement from $[\M]$; ii) $\rvecx_{\ell,k} = \rvecc_{\ell,W_k}$, $\forall \ell, k$, instead of $\rvecx_{\ell,k} = \rvecc_{\ell,W_k} \ind{\|\rvecc_{\ell,W_k}\|\le \n P}$. The total variation between the original measure and the new one is upper-bounded by 
\begin{align}
    &\Prob\lefto[|\widetilde \setW| < \Ka\right] + \Prob\lefto[ \exists k\in [\Ka] \colon \|\rvecc_{\ell,W_k}\|\ge \n P \right] \notag \\
    &\le 1 - \frac{\M!}{M^{\Ka}(M-\Ka)!} + \Ka\frac{\Gamma(\n/2,\n P/(2\power'))}{\Gamma(\n/2)} \label{eq:tmp534}\\
    &= p_0.
\end{align}
The inequality~\eqref{eq:tmp534} follows from the same analysis as in~\cite[App.~A-A]{Ngo2023_randomActivity}.
We consider implicitly the new measure hereafter at a cost of adding $p_0$ to the original expectation in~\eqref{eq:PUPE_frac}. Specifically, we expand this expectation as
\begin{align} \label{eq:bound_after_change_of_measure}
    P\sub{e} \le \sum_{t = 0}^{\Ka} \frac{t}{\Ka} \Prob\lefto[|\setW\sub{MD}| = |\setW\sub{FP}| = t\right] + p_0,
\end{align}
and next focus on upper-bounding $\Prob\lefto[|\setW\sub{MD}| = |\setW\sub{FP}| = t\right]$ under the new measure.

\subsection{Message Sets}
We further split the sets of misdetected messages, false-positive messages, and correctly decoded messages, as depicted in Fig.~\ref{fig:venn}. We first split $\setW\sub{MD}$ into two sets $\setW\sub{MD1}$ and $\setW\sub{MD2}$ that contain the misdetected messages of users in clusters $1$ and $2$, respectively. Similarly, we split $\setW\sub{FP}$ into $\setW\sub{FP1}$ and $\setW\sub{FP2}$. Recall that the decoded sets of messages for clusters $1$ and $2$ are $\widehat\setW_1$ and $\widehat\setW_2$, respectively. We denote $\setW\sub{AC1} = \setW_1 \cap \widehat\setW_2$, which is the set of messages that are false positives from the perspective of cluster~$2$ but happen to coincide with messages transmitted by cluster~$1$. We refer to these messages as ``accidentally corrected'' messages (hence the subscript AC). Similarly, we denote the set of accidentally corrected messages for cluster~$2$ as $\setW\sub{AC2} = \setW_2 \cap \widehat\setW_1$. Finally, we denote the sets of other correctly decoded messages as $\setW\sub{C1} = \widehat\setW_1 \setminus \setW\sub{MD1} \setminus \setW\sub{AC1}$ and $\setW\sub{C2} = \widehat\setW_2 \setminus \setW\sub{MD2} \setminus  \setW\sub{AC2}$. 
    \begin{figure}[t!]
		\centering
		\begin{tikzpicture}[thick,scale=.9, every node/.style={scale=.95}]
			\def\radius{2.15cm}
			\def\radiusB{\radius}
			\def\mycolorbox#1{\textcolor{#1}{\rule{2ex}{2ex}}}
			\colorlet{colori}{gray!80}
			\colorlet{colorii}{gray!20}
			
			\coordinate (ceni) at (0,0);
			\coordinate[xshift=.8*\radius] (cenii);
			
			\coordinate (edge1a) at (-\radiusB,0cm);
			\coordinate(edge1b) at (\radiusB,-.2cm);
			
			\coordinate (edge2a) at (\radius-\radiusB-.3cm,0cm);
			\coordinate (edge2b) at (\radius+\radiusB-.35cm,-.2cm);
			
			\draw[fill=colori,fill opacity=0.5] (ceni) circle (\radiusB);
			\draw[fill=colorii,fill opacity=0.5] (cenii) circle (\radius);
			\draw (ceni) circle (\radiusB);
			
			\draw (edge1a) to (edge2a);
			\draw (edge1b) to (edge2b);

            \draw (-.35cm,0cm) to node[midway,above=.2cm] {$\setW\sub{C1}$} (\radiusB,0cm);

            \draw (-.12cm,1cm) to node[midway,above] {$\setW\sub{AC1}$}(\radiusB-.2cm,1cm);

            \draw (-.12cm,-1cm) to node[midway,above=.2cm] {$\setW\sub{C2}$} node[midway,below] {$\setW\sub{AC2}$} (\radiusB-.2cm,-1cm);

			\draw[-latex] (-1.2*\radius,-0.6*\radius) node[below,xshift=-.4cm,text width=3cm,font=\linespread{1}\selectfont,align=center] {\small Distinct transmitted messages $\widetilde{\setW}$} -- (-0.85*\radius,-0.5*\radius);
            
			\draw[-latex] (2*\radius,-0.6*\radius) node[below,text width=2cm,font=\linespread{1}\selectfont,align=center] {\small Decoded messages $\widehat{\setW}$} -- (1.68*\radius,-0.5*\radius);
			
			\node[yshift=1.2*\radius,xshift=-1.25cm,text width=4.3cm,font=\linespread{1}\selectfont,align=center] {MDs: \\ ${\setW_{\rm MD}} = \setW\sub{MD1} \cup {\setW\sub{MD2}}$ \\ $= \widetilde{\setW} \setminus \widehat{\setW}$};
			
			\node[yshift=1.2*\radius,xshift=1.5cm,text width=4.3cm,font=\linespread{1}\selectfont,align=center] at (cenii) {\small FPs: \\ ${\setW_{\rm FP}} = {\setW_{\rm FP1}} \cup {\setW_{\rm FP2}}$ \\ $\ = \widehat{\setW} \setminus \widetilde{\setW}$};
			
			\node[yshift=.43*\radius,xshift=-.8cm,text width=2.3cm,font=\linespread{1}\selectfont,align=center] at (ceni) {\small $\setW\sub{MD1}$};
			\node[yshift=-.4*\radius,xshift=-.5*\radius,text width=1.8cm,font=\linespread{1}\selectfont,align=center] at (ceni) {\small $\setW\sub{MD2}$};
			\node[yshift=.4*\radius,xshift=.4*\radius,text width=1.7cm,font=\linespread{1}\selectfont,align=center] at (cenii) {\small $\setW\sub{FP1}$};
			\node[yshift=-1.2cm,xshift=.4*\radius,text width=1.9cm,font=\linespread{1}\selectfont,align=center] at (cenii) {\small $\setW\sub{FP2}$};
		\end{tikzpicture}
		\caption{A diagram depicting the relation between the sets of messages.}
		\label{fig:venn}
	\end{figure}

Exploiting symmetry, we assume without loss of generality that $\widetilde \setW = [\Ka]$, $\widetilde \setW_1 = [\Kaf]$, and thus $\widetilde \setW_2= [\Kaf+1:\Ka]$. 
Denote the cardinality of the message sets as $t\sub{MD1} = |\setW\sub{MD1}|$, $t\sub{FP1} = |\setW\sub{FP1}|$, $t\sub{AC1} = |\setW\sub{AC1}|$, and $t\sub{AC2} = |\setW\sub{AC2}|$. Under the event $|\setW\sub{MD}| = |\setW\sub{FP}| = t$, we have that $|\setW\sub{MD2}| = t - t\sub{MD1}$ and $|\setW\sub{FP2}| = t - t\sub{FP1}$. As $\setW_1 = \setW\sub{C1} \cup \setW\sub{AC1} \cup \setW\sub{MD1}$, $\widehat\setW_1 = \setW\sub{C1} \cup \setW\sub{AC2} \cup \setW\sub{FP1}$, and $|\setW_1| = |\widehat\setW_1| = \Kaf$, we obtain that 
\begin{equation} \label{eq:t1_relation}
    t\sub{FP1} + t\sub{AC2} = t\sub{MD1} + t\sub{AC1} = \Kaf - t\sub{C1} \le \Kaf.
\end{equation}
Similarly, we have that
\begin{equation} \label{eq:t2_relation}
    t\sub{FP2} + t\sub{AC1} = t\sub{MD2} + t\sub{AC2} = \Kas - t\sub{C2} \le \Kas.
\end{equation}
%
Therefore, given $t$, $t\sub{MD1}$ is upper-bounded by both $t$ and $\Kaf$, and lower-bounded as $t\sub{MD1} = t - t\sub{MD2} \ge t - \Kas$. That is, $t\sub{MD1}$ is bounded in $\setT_t$ defined in~\eqref{eq:Tt}. Similarly, so is $t\sub{FP1}$. 
Furthermore, given $t\sub{MD1}$ and $t\sub{FP1}$, $t\sub{AC1}$ is bounded in $\setT'_{t,t\sub{MD1},t\sub{FP1}}$ defined in~\eqref{eq:Ttt}.

\subsection{Pairwise Error Event}
Using the above set definitions, we express the received signal as
\begin{align}
    \rvecy &= g_1 [c_1(\setW\sub{MD1})  + c_1(\setW\sub{AC1}) + c_1(\setW\sub{C1})] \notag \\
    &\quad + g_2 [c_2(\setW\sub{MD2})  + c_2(\setW\sub{AC2}) + c_2(\setW\sub{C2})]  + \rvecz. \label{eq:y_expression}
\end{align}
Furthermore, the sums of the codewords corresponding to the decoded message sets $(\widehat \setW_1, \widehat \setW_2)$ are expressed as
\begin{align}
    c_1(\widehat \setW_1) &= c_1(\setW\sub{FP1}) + c_1(\setW\sub{AC2}) + c_1(\setW\sub{C1}), \label{eq:hatW1_expression}\\
    c_2(\widehat \setW_2) &= c_2(\setW\sub{FP2}) + c_2(\setW\sub{AC1}) + c_2(\setW\sub{C2}). \label{eq:hatW2_expression}
\end{align}
The pairwise error event $\widetilde \setW \to \widehat \setW$ implies that 
\begin{align}
    \|\rvecy - g_1 c_1(\widehat \setW_1) - g_2 c_2(\widehat \setW_2)\| < \|\rvecz\|. \label{eq:pairwise}
\end{align}
Using~\eqref{eq:y_expression}, \eqref{eq:hatW1_expression}, and \eqref{eq:hatW2_expression}, we obtain that~\eqref{eq:pairwise} is equivalent~to
\begin{align}
    \|&\rvecz + g_1 c_1(\setW\sub{MD1}) + g_2 c_2(\setW\sub{MD2}) \notag \\
    &\quad - g_1 c_1(\setW\sub{FP1}) - g_2 c_2(\setW\sub{FP2}) \notag \\
    &\quad + g_1 c_1(\setW\sub{AC1}) - g_2 c_2(\setW\sub{AC1}) \notag \\
    &\quad - g_1 c_1(\setW\sub{AC2}) + g_2 c_2(\setW\sub{AC2})\|    \le \|\rvecz\|. \label{eq:pairwise_subsets}
\end{align}
We denote by $F(\setW_{{\rm MD}\ell}, \setW_{{\rm FP}\ell}, \setW_{{\rm AC}\ell}, \ell\in \{1,2\})$ the set of $(\setW_{{\rm MD}\ell}, \setW_{{\rm FP}\ell}, \setW_{{\rm AC}\ell}, \ell\in \{1,2\})$ such that~\eqref{eq:pairwise_subsets} holds. 

\subsection{Error-Exponent Analysis} 

By writing the event $|\setW\sub{MD}| = |\setW\sub{FP}| = t$ as a union of the pairwise error events $F(\setW_{{\rm MD}\ell}, \setW_{{\rm FP}\ell}, \setW_{{\rm AC}\ell}, \ell\in \{1,2\})$, we have that
\begin{align}
    &\Prob\lefto[|\setW\sub{MD}| = |\setW\sub{FP}| = t\right] \notag \\
    &= \Prob\Bigg[\bigcup_{t\sub{MD1} \in \setT_t} \bigcup_{t\sub{FP1} \in \setT_t} \bigcup_{t\sub{AC1} \in \setT'_{t,t\sub{MD1},t\sub{FP1}}} \notag \\
    &\qquad \bigcup_{\setW\sub{MD1} \subset \binom{[\Kaf]}{t\sub{MD1}}}  
    \bigcup_{\setW\sub{MD2} \subset \binom{[\Kaf+1:\Ka]}{t- t\sub{MD1}}}  \notag \\ 
    &\qquad 
    \bigcup_{\setW\sub{FP1} \subset \binom{[\Ka + 1:\M]}{t\sub{FP1}} } 
    \bigcup_{\setW\sub{FP2} \subset \binom{[\Ka + 1:\M] \setminus \setW\sub{FP1}}{t-t\sub{FP1}}} \notag \\ 
    &\qquad 
    \bigcup_{\setW\sub{AC1} \subset \binom{[\Kaf] \setminus \setW\sub{MD1}}{t\sub{AC1}}} 
    \bigcup_{\setW\sub{AC2} \subset \binom{[\Kaf+1 : \Ka] \setminus \setW\sub{MD2}}{t\sub{AC1} + t\sub{MD1} - t\sub{FP1}}} \notag \\
    &\qquad F(\setW_{{\rm MD}\ell}, \setW_{{\rm FP}\ell}, \setW_{{\rm AC}\ell}, \ell\in \{1,2\}) \Bigg] \\ 
    &\le \sum_{t\sub{MD1} \in \setT_t} \sum_{t\sub{FP1} \in \setT_t} \sum_{t\sub{AC1} \in \setT'_{t,t\sub{MD1},t\sub{FP1}}}  \notag \\
    &\quad \Prob\Bigg[\bigcup_{\setW\sub{MD1} \subset \binom{[\Kaf]}{t\sub{MD1}}}  
    \bigcup_{\setW\sub{MD2} \subset \binom{[\Kaf+1:\Ka]}{t- t\sub{MD1}}}  \notag \\ 
    &\qquad 
    \bigcup_{\setW\sub{FP1} \subset \binom{[\Ka + 1:\M]}{t\sub{FP1}} } 
    \bigcup_{\setW\sub{FP2} \subset \binom{[\Ka + 1:\M] \setminus \setW\sub{FP1}}{t-t\sub{FP1}}} \notag \\ 
    &\qquad 
    \bigcup_{\setW\sub{AC1} \subset \binom{[\Kaf] \setminus \setW\sub{MD1}}{t\sub{AC1}}} 
    \bigcup_{\setW\sub{AC2} \subset \binom{[\Kaf+1 : \Ka] \setminus \setW\sub{MD2}}{t\sub{AC1} + t\sub{MD1} - t\sub{FP1}}} \notag \\
    &\qquad F(\setW_{{\rm MD}\ell}, \setW_{{\rm FP}\ell}, \setW_{{\rm AC}\ell}, \ell\in \{1,2\}) \Bigg]. \label{eq:tmp555}
\end{align}

Given $c_\ell(\setW_{{\rm MD}\ell})$, $c_\ell(\setW_{{\rm FP}\ell})$, $\ell \in \{1,2\}$, and $\rvecz$, it holds for every $\lambda > 0$ 
that
\begin{align}
    &\Prob\big[F(\setW_{{\rm MD}\ell}, \setW_{{\rm FP}\ell}, \setW_{{\rm AC}\ell}, \ell\in \{1,2\}) \big] \notag \\ 
    &\le \exp(\lambda\|\rvecz\|^2) \Exop_{\setW\sub{AC1}, \setW\sub{AC2}}\big[\exp(-\lambda\|\rvecz \notag \\
    &\quad +  g_1 c_1(\setW\sub{MD1}) + g_2 c_2(\setW\sub{MD2}) \notag \\
    &\quad - g_1 c_1(\setW\sub{FP1}) - g_2 c_2(\setW\sub{FP2})  + g_1 c_1(\setW\sub{AC1}) \notag \\
    &\quad - g_2 c_2(\setW\sub{AC1})  - g_1 c_1(\setW\sub{AC2}) + g_2 c_2(\setW\sub{AC2})\|^2)\big] \label{eq:tmp515}\\ 
    &= \exp(\lambda\|\rvecz\|^2) (1+ 2\lambda(g_1^2 + g_2^2)(t\sub{AC1} + t\sub{AC2})\power')^{-\n/2} \notag \\
    &\quad \cdot \exp\big(-\lambda\|\rvecz  +  g_1 c_1(\setW\sub{MD1}) + g_2 c_2(\setW\sub{MD2}) \notag \\
    &\qquad - g_1 c_1(\setW\sub{FP1}) - g_2 c_2(\setW\sub{FP2})\|^2)\notag \\
    &\qquad \cdot (1+ 2\lambda(g_1^2 + g_2^2)(t\sub{AC1} + t\sub{AC2})\power')^{-1}\big),  \label{eq:tmp519}
\end{align}
where we have applied the Chernoff bound in Lemma~\ref{lem:Chernoff} in~\eqref{eq:tmp515} and then used~\eqref{eq:quad_Gaussian} in Lemma~\ref{lem:chi2} to compute the expectation. Note that $(g_1^2 + g_2^2)(t\sub{AC1} + t\sub{AC2})$ is equal to $\kappa\sub{AC}$ defined in~\eqref{eq:kappaAC_joint}. Next, we apply Gallager's $\rho$-trick in Lemma~\ref{lem:Gallager} to get that, given $c_\ell(\setW_{{\rm MD}\ell})$, $c_\ell(\setW_{{\rm FP}\ell})$, $\ell \in \{1,2\}$, and $\rvecz$, it holds for every $\rho_1 \in [0,1]$ that 
\begin{align}
    &\Prob\Bigg[
    \bigcup_{\setW\sub{AC1} \subset \binom{[\Kaf] \setminus \setW\sub{MD1}}{t\sub{AC1}}} 
    \bigcup_{\setW\sub{AC2} \subset \binom{[\Kaf+1 : \Ka] \setminus \setW\sub{MD2}}{t\sub{AC1} + t\sub{MD1} - t\sub{FP1}}} \notag \\
    &\qquad F(\setW_{{\rm MD}\ell}, \setW_{{\rm FP}\ell}, \setW_{{\rm AC}\ell}, \ell\in \{1,2\}) \Bigg] \notag \\
    &\le \binom{\Kaf-t\sub{MD1}}{t\sub{AC1}}^{\rho_1} \binom{\Ka - \Kaf - t + t\sub{MD1}}{t\sub{AC1} + t\sub{MD1} - t\sub{FP1}}^{\rho_1}
    \notag \\ 
    &\quad \cdot (1+ 2\lambda \kappa\sub{AC} \power')^{-\n \rho_1/2} \notag \\ 
    &\quad \cdot \exp\Big(\mu_1 \Big(\|\rvecz\|^2 - \|\rvecz  +  g_1 c_1(\setW\sub{MD1}) + g_2 c_2(\setW\sub{MD2}) \notag \\
    &\qquad - g_1 c_1(\setW\sub{FP1}) - g_2 c_2(\setW\sub{FP2})\|^2\Big)\notag   \Big)
\end{align} 
where $\mu_1$ is defined by~\eqref{eq:mu1}. 
By taking the expectation over $\setW\sub{FP1}$ and $\setW\sub{FP2}$ using~\eqref{eq:quad_Gaussian} in Lemma~\ref{lem:chi2}, we obtain that, for given $c_\ell(\setW_{{\rm MD}\ell})$, $\ell \in \{1,2\}$, and given $\rvecz$,
\begin{align}
    &\Prob\Bigg[
    \bigcup_{\setW\sub{AC1} \subset \binom{[\Kaf] \setminus \setW\sub{MD1}}{t\sub{AC1}}} 
    \bigcup_{\setW\sub{AC2} \subset \binom{[\Kaf+1 : \Ka] \setminus \setW\sub{MD2}}{t\sub{AC1} + t\sub{MD1} - t\sub{FP1}}} \notag \\
    &\qquad F(\setW_{{\rm MD}\ell}, \setW_{{\rm FP}\ell}, \setW_{{\rm AC}\ell}, \ell\in \{1,2\}) \Bigg]
    \notag \\
    &\le \binom{\Kaf-t\sub{MD1}}{t\sub{AC1}}^{\rho_1} \binom{\Ka - \Kaf - t + t\sub{MD1}}{t\sub{AC1} + t\sub{MD1} - t\sub{FP1}}^{\rho_1} \notag \\ 
    &\quad \cdot \exp\Big(\lambda\rho_1 \|\rvecz\|^2 \!-\! \mu_1 \|\rvecz  +  g_1 c_1(\setW\sub{MD1}) + g_2 c_2(\setW\sub{MD2})\|^2 \notag \\
    &\qquad \cdot (1+ 2 \mu_1 (g_1^2 t\sub{FP1} + g_2^2t\sub{FP2})\power')^{-1}-\n a_1 \Big)
 \end{align}
 where 
 \begin{align}
     a_1 &= \frac{\rho_1}{2}
\ln(1+ 2\lambda\kappa\sub{AC}\power') \notag\\
&\quad +\frac{1}{2}\ln\big(1 + 2\mu_1 (g_1^2 t\sub{FP1} + g_2^2t\sub{FP2}) \power'\big).
\end{align}
Note that $g_1^2 t\sub{FP1} + g_2^2t\sub{FP2}$ is equal to $\kappa\sub{FP}$ defined in~\eqref{eq:kappaFP}.
Now we apply Gallager's $\rho$-trick again to obtain that, given~$c_\ell(\setW_{{\rm MD}\ell})$, $c_\ell(\setW_{{\rm FP}\ell})$, $\ell \in \{1,2\}$, given $\rvecz$ and for every $\rho_2 \in [0,1]$, 
\begin{align}
    &\Prob\Bigg[
    \bigcup_{\setW\sub{FP1} \subset \binom{[\Ka + 1:\M]}{t\sub{FP1}} } 
    \bigcup_{\setW\sub{FP2} \subset \binom{[\Ka + 1:\M] \setminus \setW\sub{FP1}}{t-t\sub{FP1}}} \notag \\ 
    &\qquad 
    \bigcup_{\setW\sub{AC1} \subset \binom{[\Kaf] \setminus \setW\sub{MD1}}{t\sub{AC1}}} 
    \bigcup_{\setW\sub{AC2} \subset \binom{[\Kaf+1 : \Ka] \setminus \setW\sub{MD2}}{t\sub{AC1} + t\sub{MD1} - t\sub{FP1}}} \notag \\
    &\qquad F(\setW_{{\rm MD}\ell}, \setW_{{\rm FP}\ell}, \setW_{{\rm AC}\ell}, \ell\in \{1,2\}) \Bigg] \\ 
    &\le \binom{M-\Ka}{t\sub{FP1}}^{\rho_2} \binom{M-\Ka -t\sub{FP1}}{t - t\sub{FP1}}^{\rho_2} \notag \\ 
    &\quad \cdot  \binom{\Kaf-t\sub{MD1}}{t\sub{AC1}}^{\rho_1 \rho_2} \binom{\Ka - \Kaf - t + t\sub{MD1}}{t\sub{AC1} + t\sub{MD1} - t\sub{FP1}}^{\rho_1 \rho_2} \notag \\ 
    &\quad \cdot \Exop_{c_1(\setW\sub{MD1}), c_2(\setW\sub{MD2})}\Big[\exp\Big(\lambda \rho_1 \rho_2 \|\rvecz\|^2 \notag \\
    &\qquad - \mu_2 \|\rvecz  +  g_1 c_1(\setW\sub{MD1}) + g_2 c_2(\setW\sub{MD2})\|^2  - \n \rho_2a_1\Big)\Big] \\
    &= \binom{M-\Ka}{t\sub{FP1}}^{\rho_2} \binom{M-\Ka -t\sub{FP1}}{t - t\sub{FP1}}^{\rho_2} \notag \\ 
    &\quad \cdot \binom{\Kaf-t\sub{MD1}}{t\sub{AC1}}^{\rho_1 \rho_2} \binom{\Ka - \Kaf - t + t\sub{MD1}}{t\sub{AC1} + t\sub{MD1} - t\sub{FP1}}^{\rho_1 \rho_2} \notag \\ 
    &\quad \cdot (1 + 2\mu_2(g_1^2 t\sub{MD1} + g_2^2 t\sub{MD2})\power')^{-\n/2} \notag \\ 
    &\quad \cdot \exp\Big(  \Big(\lambda \rho_1\rho_2 -\frac{\mu_2}{1+ 2\mu_2(g_1^2 t\sub{MD1} + g_2^2 t\sub{MD2})\power'}\Big) \|\rvecz\|^2 \notag \\
    &\qquad - N\rho_2 a_1\Big) \\ 
    &= \binom{M-\Ka}{t\sub{FP1}}^{\rho_2} \binom{M-\Ka -t\sub{FP1}}{t - t\sub{FP1}}^{\rho_2} \notag \\ 
    &\quad \cdot \binom{\Kaf-t\sub{MD1}}{t\sub{AC1}}^{\rho_1 \rho_2} \binom{\Ka - \Kaf - t + t\sub{MD1}}{t\sub{AC1} + t\sub{MD1} - t\sub{FP1}}^{\rho_1 \rho_2} \notag \\ 
    &\quad \cdot \exp( b\|\rvecz\|^2  - \n\rho_2 a_1 - \n a_2 ),
\end{align}
where we used~\eqref{eq:quad_Gaussian} in Lemma~\ref{lem:chi2} to obtain the first equality, and we define $\mu_2$ by~\eqref{eq:mu2} and
\begin{align}
   b &=\lambda \rho_1 \rho_2-\frac{\mu_2}{1+ 2\mu_2(g_1^2 t\sub{MD1} + g_2^2 t\sub{MD2})\power'},\\
   a_2 &=\frac{1}{2}\ln\big(1+ 2\mu_2(g_1^2 t\sub{MD1} + g_2^2 t\sub{MD2})\power'\big).
\end{align}
Note that $g_1^2 t\sub{MD1} + g_2^2 t\sub{MD2}$ is equal to $\kappa\sub{MD}$ defined in~\eqref{eq:kappaMD}. 
We now apply Gallager's $\rho$-trick the third time to obtain that, given $\rvecz$ and for every $\rho_3 \in [0,1]$,
\begin{align}
    &= \Prob\Bigg[\bigcup_{\setW\sub{MD1} \subset \binom{[\Kaf]}{t\sub{MD1}}}  
    \bigcup_{\setW\sub{MD2} \subset \binom{[\Kaf+1:\Ka]}{t- t\sub{MD1}}}  \notag \\ 
    &\qquad 
    \bigcup_{\setW\sub{FP1} \subset \binom{[\Ka + 1:\M]}{t\sub{FP1}} } 
    \bigcup_{\setW\sub{FP2} \subset \binom{[\Ka + 1:\M] \setminus \setW\sub{FP1}}{t-t\sub{FP1}}} \notag \\ 
    &\qquad 
    \bigcup_{\setW\sub{AC1} \subset \binom{[\Kaf] \setminus \setW\sub{MD1}}{t\sub{AC1}}} 
    \bigcup_{\setW\sub{AC2} \subset \binom{[\Kaf+1 : \Ka] \setminus \setW\sub{MD2}}{t\sub{AC1} + t\sub{MD1} - t\sub{FP1}}} \notag \\
    &\qquad F(\setW_{{\rm MD}\ell}, \setW_{{\rm FP}\ell}, \setW_{{\rm AC}\ell}, \ell\in \{1,2\}) \Bigg] \\ 
    &\le \binom{\Kaf}{t\sub{MD1}}^{\rho_3} \binom{\Kas}{t - t\sub{MD1}}^{\rho_3}  \notag \\
    &\quad \cdot \binom{M-\Ka}{t\sub{FP1}}^{\rho_2\rho_3} \binom{M-\Ka -t\sub{FP1}}{t - t\sub{FP1}}^{\rho_2\rho_3} \notag \\ 
    &\quad \cdot  \binom{\Kaf-t\sub{MD1}}{t\sub{AC1}}^{\rho_1 \rho_2\rho_3} \binom{\Ka - \Kaf - t + t\sub{MD1}}{t\sub{AC1} + t\sub{MD1} - t\sub{FP1}}^{\rho_1 \rho_2\rho_3} \notag \\ 
    &\quad \cdot \Exop_{\rvecz}[\exp( b \|\rvecz\|^2  - \n\rho_2 \rho_3 a) - \n \rho_3 a_1)] \\ 
    &= \binom{\Kaf}{t\sub{MD1}}^{\rho_3} \binom{\Kas}{t - t\sub{MD1}}^{\rho_3}  \notag \\
    &\quad \cdot \binom{M-\Ka}{t\sub{FP1}}^{\rho_2\rho_3} \binom{M-\Ka -t\sub{FP1}}{t - t\sub{FP1}}^{\rho_2\rho_3} \notag \\ 
    &\quad \cdot  \binom{\Kaf-t\sub{MD1}}{t\sub{AC1}}^{\rho_1 \rho_2\rho_3} \binom{\Ka - \Kaf - t + t\sub{MD1}}{t\sub{AC1} + t\sub{MD1} - t\sub{FP1}}^{\rho_1 \rho_2\rho_3} \notag \\ 
    &\quad \cdot (1 - 2 b \rho_3)^{-\n/2} \exp(-\n \rho_2\rho_3 a_1 - \n \rho_3 a_2), \label{eq:tmp785}
\end{align}
where in the last equality, we use~\eqref{eq:quad_Gaussian} in Lemma~\ref{lem:chi2} to compute the expectation over $\rvecz$. 

We conclude that $\Prob{|\setW\sub{MD}| = |\setW\sub{FP}| = t}$ is upper-bounded by the right-hand side of~\eqref{eq:tmp785}.
Finally, by substituting this bound into~\eqref{eq:bound_after_change_of_measure}, we complete the proof.

\section{Proof of Corollary~\ref{cor:bound_separate}} \label{proof:bound_separate}
When $\setC_1 = \setC_2$, we have that $c_1(\setW) = c_2(\setW)$ for a given message set $\setW$. The proof follows the same steps as in Appendix~\ref{proof:bound}, except that, in~\eqref{eq:tmp519}, the variance of elements of $g_1 c_1(\setW\sub{AC1}) - g_2 c_2(\setW\sub{AC1}) - g_1 c_1(\setW\sub{AC2}) + g_2 c_2(\setW\sub{AC2}) = (g_1 - g_2) (c(\setW\sub{AC1}) + c(\setW\sub{AC2}))$ is given by $\kappa\sub{AC} \power'$ with $\kappa\sub{AC}$ defined in~\eqref{eq:kappaAC_separate}.

\section{Proof of Theorem~\ref{th:bound_IC}} \label{proof:bound_IC}
By performing the same change of measure as in Appendix~\ref{proof:bound}, we obtain the bound~\eqref{eq:bound_after_change_of_measure}. We next focus on upper-bounding $\Prob{|\setW\sub{MD}| = |\setW\sub{FP}| = t}$ under the new measure and for the interference-cancellation decoder.  

\subsection{Pairwise Error Event}
For this decoder, the pairwise error event $\widetilde \setW \to \widehat \setW$ implies that
\begin{align}
    \|\rvecy - g_1 c(\widehat \setW_1)\| &\le \|\rvecy - g_1 c(\setW_1)\|, \label{eq:pairwise_IC_1} \\
    \|\rvecy \!-\! g_1 c(\widehat \setW_1) \!-\! g_2 c(\widehat \setW_2)\| &\le \|\rvecy \!-\! g_1 c(\widehat \setW_1) \!-\! g_2 c(\setW_2)\| \label{eq:pairwise_IC_2}
\end{align}
where we recall that $c(\setW) = \sum_{w\in\setW} \rvecc_{w}$.
Splitting message sets as in Appendix~\ref{proof:bound}, we express~\eqref{eq:pairwise_IC_1} and~\eqref{eq:pairwise_IC_2} as
\begin{align}
    &\|\underbrace{\rvecz + g_1 c(\setW\sub{MD1}) + g_2 c(\setW\sub{MD2}) - g_1 c(\setW\sub{FP1}) + g_2 c(\setW\sub{C2})}_{=\rveca} \notag \\
    &~ + g_1 c(\setW\sub{AC1}) + (g_2- g_1) c(\setW\sub{AC2})\| \notag \\
    &\quad \le \|\underbrace{\rvecz + g_2 c(\setW\sub{C2}) + g_2 c(\setW\sub{MD2})}_{=\rvecb} +  g_2 c(\setW\sub{AC2})\|, \label{eq:pairwise_IC_equiv_1} \\ 
    &\|\underbrace{\rvecz + g_1 c(\setW\sub{MD1}) + g_2 c(\setW\sub{MD2}) - g_1 c(\setW\sub{FP1}) - g_2 c(\setW\sub{FP2})}_{=\rvecc} \notag \\
    &~ + (g_1 - g_2) (c(\setW\sub{AC1}) - c(\setW\sub{AC2})) \| \notag \\
    &\quad \le \|\underbrace{\rvecz + g_1 c(\setW\sub{MD1})  - g_1 c(\setW\sub{FP1})}_{=\rvecd} \notag \\
    &\qquad + g_1 c(\setW\sub{AC1}) - g_1 c(\setW\sub{AC2})\|,  \label{eq:pairwise_IC_equiv_2}
\end{align}
where we also defined the vectors $\rveca$, $\rvecb$, $\rvecc$, $\rvecd$ to make the notation more compact in the following.
We denote by $F\sub{IC}(\setW_{{\rm MD}\ell}, \setW_{{\rm FP}\ell}, \setW_{{\rm AC}\ell}, \ell\in \{1,2\})$ the set of $(\setW_{{\rm MD}\ell}, \setW_{{\rm FP}\ell}, \setW_{{\rm AC}\ell}, \ell\in \{1,2\})$ such that~\eqref{eq:pairwise_IC_equiv_1} and~\eqref{eq:pairwise_IC_equiv_2} hold. Note that $\setW\sub{C2}$ is fully determined given $\widetilde\setW = [\Kaf + 1 : \Ka]$, $\setW\sub{MD2}$, and $\setW\sub{AC2}$.
By writing the event $|\setW\sub{MD}| = |\setW\sub{FP}| = t$ as a union of the pairwise error events, we obtain a similar bound as in~\eqref{eq:tmp555} with $F(\setW_{{\rm MD}\ell}, \setW_{{\rm FP}\ell}, \setW_{{\rm AC}\ell}, \ell\in \{1,2\})$ replaced by $F\sub{IC}(\setW_{{\rm MD}\ell}, \setW_{{\rm FP}\ell}, \setW_{{\rm AC}\ell}, \ell\in \{1,2\})$.

\subsection{Error Exponent Analysis}
We apply the Chernoff joint bound in Corollary~\ref{cor:Chernoff_joint} to bound the pairwise error probability as
\begin{align}
    &\Prob\lefto[F\sub{IC}(\setW_{{\rm MD}\ell}, \setW_{{\rm FP}\ell}, \setW_{{\rm AC}\ell}, \ell\in \{1,2\})\right] \notag \\
    &= \Prob\Big[\|\rveca + g_1 c(\setW\sub{AC1}) + (g_2- g_1) c(\setW\sub{AC2})\|^2 \notag \\
    &\qquad ~- \|\rvecb +  g_2 c(\setW\sub{AC2})\|^2 \le 0, \notag \\ 
    &\qquad \|\rvecc  + (g_1 - g_2) (c(\setW\sub{AC1}) - c(\setW\sub{AC2})) \|^2 \notag \\
    &\qquad~ - \|\rvecd + g_1 c(\setW\sub{AC1}) - g_1 c(\setW\sub{AC2})\|^2 \le 0 \Big] \\
    &\le \Exop\Big[\exp\Big(\!-\lambda_1(\|\rveca + g_1 c(\setW\sub{AC1}) + (g_2- g_1) c(\setW\sub{AC2})\|^2 \notag \\
    &\qquad \quad- \|\rvecb +  g_2 c(\setW\sub{AC2})\|^2) \notag \\ 
    &\qquad~ -\lambda_2(\|\rvecc  + (g_1 - g_2) (c(\setW\sub{AC1}) - c(\setW\sub{AC2})) \|^2 \notag \\
    &\qquad\quad - \|\rvecd + g_1 c(\setW\sub{AC1}) - g_1 c(\setW\sub{AC2})\|^2) \Big)\Big] \label{eq:tmp1370}
\end{align}

for all $\lambda_1 > 0$ and $\lambda_2 > 0$.  

\subsubsection{Expectation over $c(\setW\sub{AC1})$ and $c(\setW\sub{AC2})$ and the First Gallager-$\rho$ Trick}
The exponent in~\eqref{eq:tmp1370} can be written as  
\begin{multline}
    [\tp{c(\setW\sub{AC1})} \ \tp{c(\setW\sub{AC2})}] (\matU\sub{AC} \otimes \matidentity_\n) \begin{bmatrix}
        c(\setW\sub{AC1}) \\ c(\setW\sub{AC2})
    \end{bmatrix} \\
    + 2 \tp{\rvecv}\sub{AC} \begin{bmatrix}
        c(\setW\sub{AC1}) \\ c(\setW\sub{AC2})
    \end{bmatrix}
    + R\sub{AC}
\end{multline}
where 
\begin{align}
    \matU\sub{AC} &= \begin{bmatrix}
        u\sub{AC}\of{1} & u\sub{AC}\of{2} \\
        u\sub{AC}\of{2} & u\sub{AC}\of{3}
    \end{bmatrix}, \\
    u\sub{AC}\of{1} &= -\lambda_1 g_1^2 -\lambda_2 g_2(g_2 - 2 g_1), \\
    u\sub{AC}\of{2} &= \lambda_1 g_1 (g_1 - g_2) + \lambda_2 g_2 (g_2 - 2 g_1), \\
    u\sub{AC}\of{3} &= -\lambda_1 g_1 (g_1 - 2g_2) - \lambda_2 g_2 (g_2 - 2g_1), \\
    \rvecv\sub{AC} &= \!\begin{bmatrix}
        -\lambda_1 g_1 \rveca - \lambda_2 (g_1 - g_2) \rvecc + \lambda_2 g_1 \rvecd \\ 
        \lambda_1(g_1 \!-\! g_2) \rveca +\! \lambda_1 g_2 \rvecb +\! \lambda_2 (g_1 \!-\! g_2) \rvecc -\! \lambda_2 g_1 \rvecd
    \end{bmatrix}\!, \\
    R\sub{AC} &= -\lambda_1 (\|\rveca\|^2 - \|\rvecb\|^2) - \lambda_2(\|\rvecc\|^2 - \|\rvecd\|^2).
\end{align}
Note that $\begin{bmatrix}
        c(\setW\sub{AC1}) \\ c(\setW\sub{AC2})
    \end{bmatrix} \sim \normal(\veczero, \matSigma\sub{AC} \otimes \matidentity_{\n})$ with $\matSigma\sub{AC} = \diag(t\sub{AC1} \power', t\sub{AC2} \power')$.
We apply~\eqref{eq:quad_Gaussian_zeroMean} in Lemma~\ref{lem:chi2} to compute the expectation in~\eqref{eq:tmp1370} over $c(\setW\sub{AC1})$, $c(\setW\sub{AC2})$ as
\begin{align}
    &\det(\matidentity_2 - 2 \matSigma\sub{AC} \matU\sub{AC})^{-\n/2} \notag \\
            &\quad \cdot\exp\big(2 \tp{\rvecv}\sub{AC}[(\matSigma\sub{AC}^{-1} - 2 \matU\sub{AC})^{-1} \otimes \matidentity_\n] \rvecv\sub{AC}  + R\sub{AC}\big)
\end{align}
where we require $\lambda_1, \lambda_2$ to satisfy that $\matSigma\sub{AC}^{-1} \succ 2 \matU\sub{AC}$.

 Now, we apply Gallager's $\rho$-trick in Lemma~\ref{lem:Gallager} to get that,
 given $c(\setW_{{\rm MD}\ell})$, $c(\setW_{{\rm FP}\ell})$, $\ell \in \{1,2\}$, $c(\setW\sub{C2})$, and $\rvecz$, it holds for every $\rho_1 \in [0,1]$ that 
\begin{align}
    &\Prob\Bigg[
    \bigcup_{\setW\sub{AC1} \subset \binom{[\Kaf] \setminus \setW\sub{MD1}}{t\sub{AC1}}} 
    \bigcup_{\setW\sub{AC2} \subset \binom{[\Kaf+1 : \Ka] \setminus \setW\sub{MD2}}{t\sub{AC1} + t\sub{MD1} - t\sub{FP1}}} \notag \\
    &\qquad F\sub{IC}(\setW_{{\rm MD}\ell}, \setW_{{\rm FP}\ell}, \setW_{{\rm AC}\ell}, \ell\in \{1,2\}) \Bigg] \notag \\
    &\le \binom{\Kaf-t\sub{MD1}}{t\sub{AC1}}^{\rho_1} \binom{\Ka - \Kaf - t + t\sub{MD1}}{t\sub{AC1} + t\sub{MD1} - t\sub{FP1}}^{\rho_1}
    \notag \\ 
    &\quad \cdot \det(\matidentity_2 - 2 \matSigma\sub{AC} \matU\sub{AC})^{-\n\rho_1/2} \notag \\
            &\quad \cdot\exp\big(2 \rho_1 \tp{\rvecv}\sub{AC}[(\matSigma\sub{AC}^{-1} - 2 \matU\sub{AC})^{-1}  \otimes \matidentity_\n] \rvecv\sub{AC}  + \rho_1 R\sub{AC}  \big). \label{eq:tmp1548}
\end{align} 

\subsubsection{Expectation over $c(\setW\sub{FP1})$ and $c(\setW\sub{FP2})$ and the Second Gallager-$\rho$ Trick}
Next, we denote
\begin{align}
    \rvece &= \rvecz + g_1 c(\setW\sub{MD1}) + g_2 c(\setW\sub{MD2}) + g_2 c(\setW\sub{C2}), \\
    \rvecf &= \rvecz + g_1 c(\setW\sub{MD1}) + g_2 c(\setW\sub{MD2}), \\
    \rvecg &= \rvecz + g_1 c(\setW\sub{MD1}).
\end{align}
It follows that $\rveca = \rvece - g_1 c(\setW\sub{FP1})$, $\rvecc = \rvecf - g_1 c(\setW\sub{FP1}) - g_2 c(\setW\sub{FP2})$, and $\rvecd = \rvecg - g_1 c(\setW\sub{FP1})$. We can express the exponent in~\eqref{eq:tmp1548} as
\begin{align}
    &2 \rho_1 \tp{\rvecv}\sub{AC}[(\matSigma\sub{AC}^{-1} - 2 \matU\sub{AC})^{-1}  \otimes \matidentity_\n] \rvecv\sub{AC}  + \rho_1 R\sub{AC} \notag \\
    &= [\tp{c(\setW\sub{FP1})} \ \tp{c(\setW\sub{FP2})}] (\matU\sub{FP} \otimes \matidentity_\n) \begin{bmatrix}
        c(\setW\sub{FP1}) \\ c(\setW\sub{FP2})
    \end{bmatrix} \notag \\
    &\quad + 2 \tp{\rvecv}\sub{FP} \begin{bmatrix}
        c(\setW\sub{FP1}) \\ c(\setW\sub{FP2})
    \end{bmatrix}
    + R\sub{FP}
\end{align}
where
\begin{align}
    \matU\sub{FP} &= 2 \rho_1  \matLambda\sub{FP} \matQ\sub{FP} - \rho_1 \!\begin{bmatrix}
        \lambda_1 g_1^2 & \!\!\lambda_2g_1 g_2 \\ \lambda_2g_1 g_2\!\! & \lambda_2 g_2^2 
    \end{bmatrix}, \\
    \rvecv\sub{FP} &= 2\rho_1 [\matLambda\sub{FP} \otimes \matidentity_\n] \rvecq\sub{FP} \notag \\
    &\quad + \rho_1 \begin{bmatrix}
        \lambda_1 g_1 \rvece + \lambda_2 g_1 \rvecf -\lambda_2 g_1 \rvecg \\ \lambda_2 g_2 \rvecf,
    \end{bmatrix}, \\
    R\sub{FP} &= 2\rho_1 \tp{\rvecq}\sub{FP} \big[(\matSigma\sub{AC}^{-1} - 2 \matU\sub{AC})^{-1} \otimes \matidentity_\n\big] \rvecq\sub{FP} \notag \\
    &\quad - \rho_1 \lambda_1 (\|\rvece\|^2 - \|\rvecb\|^2)  - \rho_1 \lambda_2 (\|\rvecf\|^2 - \|\rvecg\|^2),
\end{align}
with
\begin{align}
    \matQ\sub{FP} &= \begin{bmatrix}
        \lambda_1 g_1^2 - \lambda_2 g_1 g_2 & \lambda_2 g_2(g_1 - g_2) \\
        -\lambda_1 g_1 (g_1 - g_2) + \lambda_2 g_1 g_2 & -\lambda_2 g_2 (g_1 - g_2)
    \end{bmatrix}\!, \\ 
    \rvecq\sub{FP} &= \begin{bmatrix}
        -\lambda_1 g_1 \rvece - \lambda_2(g_1 - g_2) \rvecf + \lambda_2 g_1 \rvecg \\
        \lambda_1 (g_1 \!-\! g_2) \rvece + \lambda_1 g_2 \rvecb + \lambda_2(g_1 \!-\! g_2) \rvecf - \lambda_2 g_1 \rvecg
    \end{bmatrix}\!, \\
    \matLambda\sub{FP} &= \tp{\matQ}\sub{FP} (\matSigma\sub{AC}^{-1} \!-\! 2 \matU\sub{AC})^{-1}.
\end{align}
Note that $\begin{bmatrix}
        c(\setW\sub{FP1}) \\ c(\setW\sub{FP2})
    \end{bmatrix} \sim \normal(\veczero, \matSigma\sub{FP} \otimes \matidentity_{\n})$ with $\matSigma\sub{FP} = \diag(t\sub{FP1} \power', t\sub{FP2} \power')$.
By applying~\eqref{eq:quad_Gaussian_zeroMean} in Lemma~\ref{lem:chi2} to compute the expectation of~\eqref{eq:tmp1548} over $c(\setW\sub{FP1}), c(\setW\sub{FP2})$, we obtain that, given~$c(\setW_{{\rm MD}\ell})$, $\ell \in \{1,2\}$, $c(\setW\sub{C2})$, and $\rvecz$, 
\begin{align}
    &\Prob\Bigg[
    \bigcup_{\setW\sub{AC1} \subset \binom{[\Kaf] \setminus \setW\sub{MD1}}{t\sub{AC1}}} 
    \bigcup_{\setW\sub{AC2} \subset \binom{[\Kaf+1 : \Ka] \setminus \setW\sub{MD2}}{t\sub{AC1} + t\sub{MD1} - t\sub{FP1}}} \notag \\
    &\qquad F\sub{IC}(\setW_{{\rm MD}\ell}, \setW_{{\rm FP}\ell}, \setW_{{\rm AC}\ell}, \ell\in \{1,2\}) \Bigg] \notag \\
    &\le \binom{\Kaf-t\sub{MD1}}{t\sub{AC1}}^{\rho_1} \binom{\Ka - \Kaf - t + t\sub{MD1}}{t\sub{AC1} + t\sub{MD1} - t\sub{FP1}}^{\rho_1}
    \notag \\ 
    &\quad \cdot \det(\matidentity_2 - 2 \matSigma\sub{AC} \matU\sub{AC})^{-\n\rho_1/2} \notag \\
            &\quad \cdot \det(\matidentity_2 - 2 \matSigma\sub{FP} \matU\sub{FP})^{-\n/2} \notag \\
            &\quad \cdot \exp\big(2 \tp{\rvecv}\sub{FP}[(\matSigma\sub{FP}^{-1} - 2  \matU\sub{FP})^{-1} \otimes \matidentity_\n] \rvecv\sub{FP}  + R\sub{FP}\big)
\end{align}
under the condition that $\matSigma\sub{FP}^{-1} \succ 2 \matU\sub{FP}$.

Now, we apply Gallager's $\rho$-trick again to get that, given~$c(\setW_{{\rm MD}\ell})$, $\ell \in \{1,2\}$, $c(\setW\sub{C2})$, and $\rvecz$, and for every $\rho_2 \in [0,1]$, it holds that
\begin{align}
    &\Prob\Bigg[
    \bigcup_{\setW\sub{FP1} \subset \binom{[\Ka + 1:\M]}{t\sub{FP1}} } 
    \bigcup_{\setW\sub{FP2} \subset \binom{[\Ka + 1:\M] \setminus \setW\sub{FP1}}{t-t\sub{FP1}}} \notag \\ 
    &\qquad 
    \bigcup_{\setW\sub{AC1} \subset \binom{[\Kaf] \setminus \setW\sub{MD1}}{t\sub{AC1}}} 
    \bigcup_{\setW\sub{AC2} \subset \binom{[\Kaf+1 : \Ka] \setminus \setW\sub{MD2}}{t\sub{AC1} + t\sub{MD1} - t\sub{FP1}}} \notag \\
    &\qquad F\sub{IC}(\setW_{{\rm MD}\ell}, \setW_{{\rm FP}\ell}, \setW_{{\rm AC}\ell}, \ell\in \{1,2\}) \Bigg] \\ 
    &\le \binom{M-\Ka}{t\sub{FP1}}^{\rho_2} \binom{M-\Ka -t\sub{FP1}}{t - t\sub{FP1}}^{\rho_2} \notag \\ 
    &\quad \cdot  \binom{\Kaf-t\sub{MD1}}{t\sub{AC1}}^{\rho_1 \rho_2} \binom{\Ka - \Kaf - t + t\sub{MD1}}{t\sub{AC1} + t\sub{MD1} - t\sub{FP1}}^{\rho_1 \rho_2} \notag \\ 
    &\quad \cdot \det(\matidentity_2 - 2 \matSigma\sub{AC} \matU\sub{AC})^{-\n\rho_1\rho_2/2} \notag \\
            &\quad \cdot \det(\matidentity_2 - 2 \matSigma\sub{FP} \matU\sub{FP})^{-\n\rho_2/2} \notag \\
            &\quad \cdot \exp\big(2\rho_2 \tp{\rvecv}\sub{FP}[(\matSigma\sub{FP}^{-1} - 2  \matU\sub{FP})^{-1} \otimes \matidentity_\n] \rvecv\sub{FP}  + \rho_2 R\sub{FP}\big). \label{eq:tmp1829}
\end{align}

\subsubsection{Expectation over $c(\setW\sub{MD1})$ and $c(\setW\sub{MD2})$ and the Third Gallager-$\rho$ Trick}

We express the exponent in~\eqref{eq:tmp1829} as
\begin{align}
    &2\rho_2 \tp{\rvecv}\sub{FP}[(\matSigma\sub{FP}^{-1} - 2  \matU\sub{FP})^{-1} \otimes \matidentity_\n] \rvecv\sub{FP}  + \rho_2 R\sub{FP} \notag \\
    &= [\tp{c(\setW\sub{MD1})} \ \tp{c(\setW\sub{MD2})}] (\matU\sub{MD} \otimes \matidentity_\n) \begin{bmatrix}
        c(\setW\sub{MD1}) \\ c(\setW\sub{MD2})
    \end{bmatrix} \\
    &\quad + 2 \tp{\rvecv}\sub{MD} \begin{bmatrix}
        c(\setW\sub{MD1}) \\ c(\setW\sub{MD2})
    \end{bmatrix}
    + R\sub{MD}
\end{align}
where
\begin{align}
    \matU\sub{MD} &= 2\rho_2 \matLambda\sub{MD} \matQ\sub{MD} + 2 \rho_2 \rho_1 \matLambda_{\matM} \matM \notag \\
    &\quad - \rho_2 \rho_1 \begin{bmatrix}
        \lambda_1 g_1^2 & (\lambda_1 + \lambda_2) g_1 g_2 \\
        (\lambda_1 + \lambda_2) g_1 g_2 & \lambda_2 g_2^2
    \end{bmatrix}, \\ 
    \matM &= \!\begin{bmatrix}
        -\lambda_1 g_1^2 + \lambda_2 g_1 g_2 & \!\!\!-\lambda_1 g_1 g_2 \!-\! \lambda_2g_2(g_1 \!-\! g_2) \\
        \lambda_1 g_1 (g_1 \!-\! g_2) -\! \lambda_2 g_1 g_2\!\!\! & \lambda_1 g_1 g_2 + \lambda_2 g_2 (g_1 \!-\! g_2)
    \end{bmatrix}\!, \\ 
    \matQ\sub{MD} &= 2\rho_1 \matLambda\sub{FP} \matM + \rho_1 \begin{bmatrix}
        \lambda_1 g_1^2 & (\lambda_1 + \lambda_2) g_1 g_2 \\
        \lambda_2 g_1 g_2 & \lambda_2 g_2^2
    \end{bmatrix}, \\ 
    \rvecv\sub{MD} &=  2\rho_2 (\matLambda\sub{MD} \otimes \matidentity_\n) \rvecq\sub{MD} + 2 \rho_2 \rho_1 (\matLambda_{\matM} \otimes \matidentity_\n) \rvecm \notag \\ 
    &\quad - \rho_2 \rho_1 \rvecr, \\ 
    \rvecq\sub{MD} &= 2\rho_1 (\matLambda\sub{FP} \otimes \matidentity_\n) \rvecm + \rho_1 \rvecr,    \\
    \rvecm &= \bigg(\begin{bmatrix}
        1 & -1  \\ -1 & 1 
    \end{bmatrix}  \otimes \matidentity_\n \bigg)  \rvecr,
    \\ 
    \rvecr &= \bigg(\begin{bmatrix}
        \lambda_1 g_1 & \lambda_1 g_1 g_2 \\
        \lambda_2 g_2 & 0
    \end{bmatrix}  \otimes \matidentity_\n \bigg)  \begin{bmatrix}
        \rvecz \\ c(\setW\sub{C2})
    \end{bmatrix}, \\ %
    R\sub{MD} &= 2 \rho_2 \tp{\rvecq}\sub{MD} [(\matSigma\sub{FP}^{-1} -2 \matU\sub{FP})^{-1}] \rvecq\sub{MD} \notag \\
    &\quad + 4 \rho_1 \rho_2 \tp{\rvecm}(\matSigma\sub{AC}^{-1} - 2\matU\sub{AC})^{-1} \rvecm,
\end{align}
with 
\begin{align}
    \matLambda\sub{MD} &=
    \tp{\matQ}\sub{MD} (\matSigma\sub{FP}^{-1} - 2\matU\sub{FP})^{-1}, \\ 
    \matLambda_{\matM} &= \tp{\matM} (\matSigma\sub{AC}^{-1} - 2\matU\sub{AC})^{-1}. 
\end{align}
Note that $\begin{bmatrix}
        c(\setW\sub{MD1}) \\ c(\setW\sub{MD2})
    \end{bmatrix} \sim \normal(\veczero, \matSigma\sub{MD} \otimes \matidentity_{\n})$ with $\matSigma\sub{MD} = \diag(t\sub{MD1} \power', t\sub{MD2} \power')$.
By applying~\eqref{eq:quad_Gaussian_zeroMean} in Lemma~\ref{lem:chi2} to compute the expectation of~\eqref{eq:tmp1829} over $c(\setW\sub{MD1}), c(\setW\sub{MD2})$, we obtain that, given~$c(\setW\sub{C2})$ and $\rvecz$, 
\begin{align}
    &\Prob\Bigg[\bigcup_{\setW\sub{FP1} \subset \binom{[\Ka + 1:\M]}{t\sub{FP1}} } 
    \bigcup_{\setW\sub{FP2} \subset \binom{[\Ka + 1:\M] \setminus \setW\sub{FP1}}{t-t\sub{FP1}}} \notag \\ 
    &\qquad 
    \bigcup_{\setW\sub{AC1} \subset \binom{[\Kaf] \setminus \setW\sub{MD1}}{t\sub{AC1}}} 
    \bigcup_{\setW\sub{AC2} \subset \binom{[\Kaf+1 : \Ka] \setminus \setW\sub{MD2}}{t\sub{AC1} + t\sub{MD1} - t\sub{FP1}}} \notag \\
    &\qquad F\sub{IC}(\setW_{{\rm MD}\ell}, \setW_{{\rm FP}\ell}, \setW_{{\rm AC}\ell}, \ell\in \{1,2\}) \Bigg] \notag \\
    &\le \binom{\Kaf-t\sub{MD1}}{t\sub{AC1}}^{\rho_1} \binom{\Ka - \Kaf - t + t\sub{MD1}}{t\sub{AC1} + t\sub{MD1} - t\sub{FP1}}^{\rho_1}
    \notag \\ 
    &\quad \cdot \det(\matidentity_2 - 2 \matSigma\sub{AC} \matU\sub{AC})^{-\n\rho_1\rho_2/2} \notag \\
    &\quad \cdot \det(\matidentity_2 - 2 \matSigma\sub{FP} \matU\sub{FP})^{-\n\rho_2/2} \notag \\
    &\quad \cdot \det(\matidentity_2 - 2 \matSigma\sub{MD} \matU\sub{MD})^{-\n/2} \notag \\
    &\quad \cdot \exp\big(2 \tp{\rvecv}\sub{MD}[(\matSigma\sub{MD}^{-1} - 2  \matU\sub{MD})^{-1} \otimes \matidentity_\n] \rvecv\sub{MD}  + R\sub{MD}\big)
\end{align}
under the condition that $\matSigma\sub{FP}^{-1} \succ 2 \matU\sub{FP}$. 

We now apply Gallager's $\rho$-trick the third time to obtain that, given $c(\setW\sub{C2})$ and $\rvecz$, and for every $\rho_3 \in [0,1]$,
\begin{align}
    &= \Prob\Bigg[\bigcup_{\setW\sub{MD1} \subset \binom{[\Kaf]}{t\sub{MD1}}}  
    \bigcup_{\setW\sub{MD2} \subset \binom{[\Kaf+1:\Ka]}{t- t\sub{MD1}}}  \notag \\ 
    &\qquad 
    \bigcup_{\setW\sub{FP1} \subset \binom{[\Ka + 1:\M]}{t\sub{FP1}} } 
    \bigcup_{\setW\sub{FP2} \subset \binom{[\Ka + 1:\M] \setminus \setW\sub{FP1}}{t-t\sub{FP1}}} \notag \\ 
    &\qquad 
    \bigcup_{\setW\sub{AC1} \subset \binom{[\Kaf] \setminus \setW\sub{MD1}}{t\sub{AC1}}} 
    \bigcup_{\setW\sub{AC2} \subset \binom{[\Kaf+1 : \Ka] \setminus \setW\sub{MD2}}{t\sub{AC1} + t\sub{MD1} - t\sub{FP1}}} \notag \\
    &\qquad F\sub{IC}(\setW_{{\rm MD}\ell}, \setW_{{\rm FP}\ell}, \setW_{{\rm AC}\ell}, \ell\in \{1,2\}) \Bigg] \\ 
    &\le \binom{\Kaf}{t\sub{MD1}}^{\rho_3} \binom{\Kas}{t - t\sub{MD1}}^{\rho_3}  \notag \\
    &\quad \cdot \binom{M-\Ka}{t\sub{FP1}}^{\rho_2\rho_3} \binom{M-\Ka -t\sub{FP1}}{t - t\sub{FP1}}^{\rho_2\rho_3} \notag \\ 
    &\quad \cdot  \binom{\Kaf-t\sub{MD1}}{t\sub{AC1}}^{\rho_1 \rho_2\rho_3} \binom{\Ka - \Kaf - t + t\sub{MD1}}{t\sub{AC1} + t\sub{MD1} - t\sub{FP1}}^{\rho_1 \rho_2\rho_3} \notag \\ 
    &\quad \cdot \det(\matidentity_2 - 2 \matSigma\sub{AC} \matU\sub{AC})^{-\n\rho_1\rho_2\rho_3/2} \notag \\
    &\quad \cdot \det(\matidentity_2 - 2 \matSigma\sub{FP} \matU\sub{FP})^{-\n\rho_2\rho_3/2} \notag \\
    &\quad \cdot \det(\matidentity_2 - 2 \matSigma\sub{MD} \matU\sub{MD})^{-\n\rho_3/2} \notag \\
    &\quad \cdot \exp\big(2 \rho_3\tp{\rvecv}\sub{MD}[(\matSigma\sub{MD}^{-1} - 2  \matU\sub{MD})^{-1} \otimes \matidentity_\n] \rvecv\sub{MD}  + \rho_3 R\sub{MD}\big). \label{eq:tmp1992}
\end{align}

\subsubsection{Expectation over $\rvecz$ and $c(\setW\sub{C2})$}
We express the exponent in~\eqref{eq:tmp1992} as
\begin{align}
    &2 \rho_3\tp{\rvecv}\sub{MD}[(\matSigma\sub{MD}^{-1} - 2  \matU\sub{MD})^{-1} \otimes \matidentity_\n] \rvecv\sub{MD}  + \rho_3 R\sub{MD} \notag \\
    &= [\tp{\rvecz} \ \tp{c(\setW\sub{C2})}] (\matU\sub{ZC} \otimes \matidentity_\n) \begin{bmatrix}
        \rvecz \\ c(\setW\sub{C2})
    \end{bmatrix}
\end{align}
where
\begin{align}
    \matU\sub{ZC} &= 2\rho_3 \tp{\matQ}\sub{ZC}(\matSigma\sub{MD}^{-1} - 2  \matU\sub{MD})^{-1} \matQ\sub{ZC} \notag \\
    &\quad + 2 \rho_3 \rho_2 \tp{\matLambda}\sub{ZCa} (\matSigma\sub{FP}^{-1} - 2\matU\sub{FP})^{-1} \matLambda\sub{ZCa} \notag \\
    &\quad + 4 \rho_3 \rho_2 \rho_1 \tp{\matLambda}\sub{ZCb}  (\matSigma\sub{AC}^{-1} - 2\matU\sub{AC})^{-1} \matLambda\sub{ZCb}
\end{align}
with
\begin{align}
    \matQ\sub{ZC} &= \rho_2 \rho_1  \bigg(-\matidentity_2 + (4\matLambda\sub{MD} \matLambda\sub{FP} + 2 \matLambda_{\matM}) \begin{bmatrix}
        -1\!\! & 1 \\ 1 & \!- 1
    \end{bmatrix} + 2 \matLambda\sub{MD} \bigg) \notag \\&\quad \cdot \begin{bmatrix}
        \lambda_1 g_1 & \lambda_1 g_1 g_2 \\
        \lambda_2 g_2 & 0
    \end{bmatrix}, \\ 
    \matLambda\sub{ZCa} &= \rho_1 \bigg(\matidentity_2 + 2 \matLambda\sub{FP} \begin{bmatrix}
        -1 & 1 \\ 1 & - 1
    \end{bmatrix}\bigg)  \begin{bmatrix}
        \lambda_1 g_1 & \lambda_1 g_1 g_2 \\
        \lambda_2 g_2 & 0
    \end{bmatrix}\!, \\
    \matLambda\sub{ZCb} &= \begin{bmatrix}
        -1 & 1 \\ 1 & - 1
    \end{bmatrix}  \begin{bmatrix}
        \lambda_1 g_1 & \lambda_1 g_1 g_2 \\
        \lambda_2 g_2 & 0
    \end{bmatrix}.
\end{align}

Note that $\begin{bmatrix}
        \rvecz \\ c(\setW\sub{C2})
    \end{bmatrix} \sim \normal(\veczero, \matSigma\sub{ZC} \otimes \matidentity_{\n})$ with $\matSigma\sub{ZC} = \diag(1, (\Kas -t\sub{MD2} - t\sub{AC2}) \power')$.
By applying~\eqref{eq:quad_Gaussian_zeroMean} in Lemma~\ref{lem:chi2} to compute the expectation of~\eqref{eq:tmp1829} over $\rvecz$ and $c(\setW\sub{C2})$, we obtain that
\begin{align}
    &= \Prob\Bigg[\bigcup_{\setW\sub{MD1} \subset \binom{[\Kaf]}{t\sub{MD1}}}  
    \bigcup_{\setW\sub{MD2} \subset \binom{[\Kaf+1:\Ka]}{t- t\sub{MD1}}}  \notag \\ 
    &\qquad 
    \bigcup_{\setW\sub{FP1} \subset \binom{[\Ka + 1:\M]}{t\sub{FP1}} } 
    \bigcup_{\setW\sub{FP2} \subset \binom{[\Ka + 1:\M] \setminus \setW\sub{FP1}}{t-t\sub{FP1}}} \notag \\ 
    &\qquad 
    \bigcup_{\setW\sub{AC1} \subset \binom{[\Kaf] \setminus \setW\sub{MD1}}{t\sub{AC1}}} 
    \bigcup_{\setW\sub{AC2} \subset \binom{[\Kaf+1 : \Ka] \setminus \setW\sub{MD2}}{t\sub{AC1} + t\sub{MD1} - t\sub{FP1}}} \notag \\
    &\qquad F\sub{IC}(\setW_{{\rm MD}\ell}, \setW_{{\rm FP}\ell}, \setW_{{\rm AC}\ell}, \ell\in \{1,2\}) \Bigg] \\ 
    &\le \binom{\Kaf}{t\sub{MD1}}^{\rho_3} \binom{\Kas}{t - t\sub{MD1}}^{\rho_3}  \notag \\
    &\quad \cdot \binom{M-\Ka}{t\sub{FP1}}^{\rho_2\rho_3} \binom{M-\Ka -t\sub{FP1}}{t - t\sub{FP1}}^{\rho_2\rho_3} \notag \\ 
    &\quad \cdot  \binom{\Kaf-t\sub{MD1}}{t\sub{AC1}}^{\rho_1 \rho_2\rho_3} \binom{\Ka - \Kaf - t + t\sub{MD1}}{t\sub{AC1} + t\sub{MD1} - t\sub{FP1}}^{\rho_1 \rho_2\rho_3} \notag \\ 
    &\quad \cdot \det(\matidentity_2 - 2 \matSigma\sub{AC} \matU\sub{AC})^{-\n\rho_1\rho_2\rho_3/2} \notag \\
    &\quad \cdot \det(\matidentity_2 - 2 \matSigma\sub{FP} \matU\sub{FP})^{-\n\rho_2\rho_3/2} \notag \\
    &\quad \cdot \det(\matidentity_2 - 2 \matSigma\sub{MD} \matU\sub{MD})^{-\n\rho_3/2} \notag \\
    &\quad \cdot \det(\matidentity_2 - 2 \matSigma\sub{ZC} \matU\sub{ZC})^{-\n/2}, \label{eq:tmp2306}
\end{align}
which is an upper-bound for $\Prob{|\setW\sub{MD}| = |\setW\sub{FP}| = t}$.
Finally, by substituting this bound into~\eqref{eq:bound_after_change_of_measure} and rearranging the terms, we complete the proof.

\end{appendices}
}
\end{document}